\pdfoutput=1

\documentclass[runningheads]{llncs}
\usepackage[T1]{fontenc}
\usepackage{graphicx}

\usepackage{amsmath}
\usepackage{amsfonts}
\usepackage{amssymb}
\usepackage{stmaryrd}

\usepackage{hyperref}
\usepackage{cleveref}
\usepackage{color}

\usepackage{xcolor}

\usepackage{bussproofs}

\usepackage{nameref}

\usepackage{thm-restate}

\usepackage{tcolorbox}

\usepackage{listings}
\usepackage{relsize}
\usepackage{algorithm}

\usepackage{stackengine}

\usepackage{booktabs}

\usepackage{wrapfig}
\usepackage{subcaption}
\usepackage{orcidlink}

\usepackage{courierten}

\newcommand{\postF}{\Xi}
\newcommand{\condF}{\mathcal{C}}
\newcommand{\parameter}{\mathfrak{P}}

\newcommand{\ldot}{\mathpunct{.}}

\newcommand{\bsem}[1]{\llparenthesis #1 \rrparenthesis}
\newcommand{\sem}[1]{\llbracket #1 \rrbracket}

\newcommand{\prog}{\mathbb{P}}
\newcommand{\progg}{\mathbb{Q}}
\newcommand{\progVars}{\mathcal{V}}
\newcommand{\nat}{\mathbb{N}}
\newcommand{\bool}{\mathbb{B}}
\newcommand{\intSet}{\mathbb{Z}}

\definecolor{programOrange}{RGB}{151, 64,7}

\newcommand{\myif}{\text{\color{programOrange}\texttt{if}}}
\newcommand{\myassume}{\text{\color{programOrange}\texttt{assume}}}

\newcommand{\mywhile}{\text{\color{programOrange}\texttt{while}}}
\newcommand{\myskip}{\text{\color{programOrange}\texttt{skip}}}

\DeclareMathOperator{\myassign}{{\color{programOrange}=}}

\DeclareMathOperator{\mysemi}{{\color{programOrange}\fatsemi}}

\newcommand{\rename}[2]{#1_{\langle #2 \rangle}}

\newcommand\xqed[1]{%
	\leavevmode\unskip\penalty9999 \hbox{}\nobreak\hfill
	\quad\hbox{#1}}
\newcommand\demo{\xqed{$\triangle$}}

\definecolor{plotRed}{HTML}{C1272D}
\definecolor{plotBlue}{HTML}{0000A7}

\newcommand{\rel}[4]{\boldsymbol{\langle} #1 \boldsymbol{\rangle} #2 \sim #3 \boldsymbol{\langle} #4 \boldsymbol{\rangle}}
\newcommand{\relBig}[4]{\Big\langle #1 \Big\rangle \, #2 \sim #3 \Big\langle #4 \Big\rangle}
\newcommand{\relBigg}[4]{\Bigg\langle #1 \Bigg\rangle  #2\sim #3  \Bigg\langle #4 \Bigg\rangle}

\newcommand{\tripH}[3]{{\{} #1 {\}} #2 {\{} #3 {\}}}
\newcommand{\tripN}[3]{ \boldsymbol{[} \,#1\, \boldsymbol{]} #2 \boldsymbol{[} \,#3\, \boldsymbol{]} }

\newcommand{\univC}{\overline{\chi_\forall} }
\newcommand{\existsC}{\overline{\chi_\exists} }

\usepackage{scalefnt}

\newcommand\ScaleExists[1]{\vcenter{\hbox{\scalefont{#1}$\exists$}}}
\newcommand\ScaleForall[1]{\vcenter{\hbox{\scalefont{#1}$\forall$}}}

\DeclareMathOperator*\bigexists{%
	\vphantom\sum
	\mathchoice{\ScaleExists{1.7}}{\ScaleExists{1.4}}{\ScaleExists{1}}{\ScaleExists{0.75}}}

\DeclareMathOperator*\bigforall{%
	\vphantom\sum
	\mathchoice{\ScaleForall{1.7}}{\ScaleForall{1.4}}{\ScaleForall{1}}{\ScaleForall{0.75}}}

\makeatletter
\DeclareRobustCommand\bigop[1]{%
	\mathop{\vphantom{\sum}\mathpalette\bigop@{#1}}\slimits@
}
\newcommand{\bigop@}[2]{%
	\vcenter{%
		\sbox\z@{$#1\sum$}%
		\hbox{\resizebox{\ifx#1\displaystyle.9\fi\dimexpr\ht\z@+\dp\z@}{!}{$\m@th#2$}}%
	}%
}
\makeatother

\newcommand{\bigoast}{\DOTSB\bigop{\oast}}

\newcommand{\HyPA}{\texttt{HyPA}}
\newcommand{\PCSAT}{\texttt{PCSat}}
\newcommand{\HyPro}{\texttt{HyPro}}
\newcommand{\tool}{\texttt{ForEx}}

\definecolor{dkcyan}{rgb}{0.1, 0.3, 0.3}
\definecolor{dkgreen}{rgb}{0,0.3,0}

\colorlet{comment-color}{black!50}

\lstdefinelanguage{code-lang}{
	keywords={let, in, match, with, when, if, then, else, elif, for, to, do, return, from, def, :=, max},
	keywordstyle=[1]\bfseries,
	morekeywords=[2]{genpp, genppLoops},
	keywordstyle=[2]\color{dkgreen},
	morekeywords=[3]{guessInvariantAndCounts,freshParameter},
	keywordstyle=[3]\color{dkcyan},
	comment=[l][\color{comment-color}]{//},
	literate=%
	{:=}{{{\bfseries:=}}}1
	{:}{{{\bfseries:}}}1
	{@}{ }1
}

\lstdefinestyle{code-style}{
	escapeinside={(*}{*)},
	basicstyle=\ttfamily\fontsize{8}{9.6}\selectfont,
	columns=fullflexible,
	commentstyle=\sffamily\color{black!50!white},
	framexleftmargin=1em,
	framexrightmargin=1ex,
	keepspaces=true,
	keywordstyle=\color{dkblue},
	mathescape,
	numbers=left,
	numberblanklines=false,
	numbersep=0.5em,
	numberstyle=\relscale{0.65}\color{gray}\ttfamily,
	showstringspaces=true,
	stepnumber=1,
	xleftmargin=1.2em,
}

\lstdefinestyle{code-style-large}{
	escapeinside={(*}{*)},
	basicstyle=\ttfamily\fontsize{8.4}{9.7}\selectfont,
	columns=fullflexible,
	commentstyle=\sffamily\color{black!50!white},
	framexleftmargin=1em,
	framexrightmargin=1ex,
	keepspaces=true,
	keywordstyle=\color{dkblue},
	mathescape,
	numbers=left,
	numberblanklines=false,
	numbersep=0.5em,
	numberstyle=\relscale{0.75}\color{gray}\ttfamily,
	showstringspaces=true,
	stepnumber=1,
	xleftmargin=1.2em,
}

\lstdefinelanguage{example-lang}{
	keywords={while,do,if, then,else},
	keywordstyle=[1]{\color{programOrange}},
	comment=[l][\color{comment-color}]{//},
	literate=%
	{=}{{{\color{programOrange}=}}}1
	{>}{{{\color{programOrange}>}}}1
	{+}{{{\color{programOrange}+}}}1
	{-}{{{\color{programOrange}-}}}1
	{*}{{{\color{programOrange}*}}}1
	{@}{ }1
}

\lstdefinestyle{example-style}{
	escapeinside={(*}{*)},
	basicstyle=\ttfamily\fontsize{9}{10.8}\selectfont,
	columns=fullflexible,
	commentstyle=\sffamily\color{black!50!white},
	framexleftmargin=0em,
	framexrightmargin=0ex,
	keepspaces=true,
	keywordstyle=\color{dkblue},
	mathescape,
	numbers=left,
	numberblanklines=false,
	showstringspaces=true,
	stepnumber=1,
	xleftmargin=0em,
	numbers=none
}

\lstnewenvironment{code}[1][]
{\small
	\lstset{
		style=code-style, 
		language=code-lang,
		#1
	}
}
{}

\lstnewenvironment{exampleCode}[1][]
{\small
	\lstset{
		style=example-style, 
		language=example-lang,
		#1
	}
}
{}

\usepackage{minibox}
\usepackage{xparse}
\newsavebox{\topprooftreebox}
\newlength{\topprooftreewidth}
\NewDocumentEnvironment{topprooftree}{m}%
{\begin{lrbox}{\topprooftreebox}\ignorespaces}%
	{\DisplayProof\end{lrbox}\begin{center}\settowidth{\topprooftreewidth}%
		{\topprooftreebox}\makebox[\topprooftreewidth]{%
			\minibox{{#1}\\\usebox{\topprooftreebox}}}\end{center}}

\newcounter{mylabelcounter}
\makeatletter
\newcommand{\labelText}[2]{%
	#1\refstepcounter{mylabelcounter}%
	\immediate\write\@auxout{%
		\string\newlabel{#2}{{1}{\thepage}{{\unexpanded{#1}}}{mylabelcounter.\number\value{mylabelcounter}}{}}%
	}%
}
\newcommand{\labelTextGrey}[2]{%
	{\color{black!50}#1}\refstepcounter{mylabelcounter}%
	\immediate\write\@auxout{%
		\string\newlabel{#2}{{1}{\thepage}{{\unexpanded{#1}}}{mylabelcounter.\number\value{mylabelcounter}}{}}%
	}%
}
\makeatother

\newif\iffullversion
\fullversiontrue

\newcommand{\ifFull}[2]{\iffullversion#1\else#2\fi}

\begin{document}
\title{Automated Software Verification of Hyperliveness}

\author{Raven Beutner \orcidlink{0000-0001-6234-5651}}
\institute{CISPA Helmholtz Center for Information Security, Germany\\
\email{raven.beutner@cispa.de}}
\authorrunning{R.~Beutner}

\maketitle
\begin{abstract}
Hyperproperties relate multiple executions of a program and are commonly used to specify security and information-flow policies.
Most existing work has focused on the verification of $k$-safety properties, i.e., properties that state that \emph{all} $k$-tuples of execution traces satisfy a given property. 
In this paper, we study the automated verification of richer properties that combine universal and existential quantification over executions. 
Concretely, we consider $\forall^k\exists^l$ properties, which state that for all $k$ executions, there exist $l$ executions that, together, satisfy a property. 
This captures important non-$k$-safety requirements, including hyperliveness properties such as generalized non-interference, opacity,  refinement, and robustness. 
We design an automated constraint-based algorithm for the verification of $\forall^k\exists^l$ properties.
Our algorithm leverages a sound-and-complete program logic and a (parameterized) strongest postcondition computation.
We implement our algorithm in a tool called \tool{} and report on encouraging experimental results. 

\end{abstract}

\section{Introduction}\label{sec:intro}

Relational properties (also called hyperproperties \cite{ClarksonS10}) move away from a traditional specification that considers all executions of a system in isolation and, instead, relate \emph{multiple} executions. 
Hyperproperties are becoming increasingly important and have shown up in various disciplines, perhaps most prominently in information-flow control.
Assume we are given a program $\prog$ with high-security input $h$, low-security input $l$, and public output $o$, and we want to formally prove that the output of $\prog$ does not leak information about $h$.
One way to ensure this is to verify that $\prog$ behaves deterministically in the low-security input $l$, i.e., if the low-security input is identical across \emph{two} executions, so is $\prog$'s output. 

The above property is a typical example of a $2$-safety property stating a requirement on all pairs of traces.
More generally, a $k$-safety property requires that \emph{all} $k$-tuples of executions, together, satisfy a given property. 
In the last decade, many approaches for the verification of $k$-safety properties have been proposed, based, e.g., on model-checking \cite{ShemerGSV19,FinkbeinerRS15,FarzanV19}, abstract interpretation \cite{MastroeniP18,KovacsSF13,AssafNSTT17,MastroeniP19}, symbolic execution \cite{FarinaCG19}, or program logics \cite{Benton04,SousaD16,DOsualdoFD22,Yang07,NagasamudramN21}.

\begin{wrapfigure}{R}{0.35\linewidth}
\newsavebox{\mybox}
\begin{lrbox}{\mybox}
\begin{exampleCode}
if ($h$ > $l$) then 
@@$o$ = $l$ + $\star_\mathbb{N}$
else 
@@$x$ = $\star_\mathbb{N}$
@@if ($x$ > $l$) then 
@@@@$o$ = $x$
@@else 
@@@@$o$ = $l$
\end{exampleCode}
\end{lrbox}
\centering
\scalebox{1}{\usebox{\mybox}}
	\vspace{-2mm}
	\caption{Example program}\label{fig:intro-example}
	\vspace{-4mm}
\end{wrapfigure}

However, for many relational properties, the implicit \emph{universal} quantification found in $k$-safety properties is too restrictive. 
Consider the simple program in \Cref{fig:intro-example} (taken from \cite{BeutnerF22}), where $\star_\mathbb{N}$ denotes the nondeterministic choice of a natural number.
This program clearly violates the $2$-safety property discussed above as the nondeterminism influences the final value of $o$.
Nevertheless, the program does not leak any information about the secret input $h$.
To see this, assume the attacker observes some fixed low-security input-output pair $(l, o)$, i.e., the attacker observes everything except the high-security input.
The key observation is that $(l, o)$ is possible for any possible high-security input, i.e., for every value of $h$, there \emph{exists} some way to resolve the nondeterminism such that $(l, o)$ is the observation made by the attacker.
This information-flow policy -- called generalized non-interference (GNI) \cite{McCullough88} -- requires a combination of universal and existential reasoning and thus cannot be expressed as a $k$-safety property.

\paragraph{FEHTs.}

In this paper, we study the automated verification of such (functional) $\forall^*\exists^*$ properties. 
Concretely, we consider specifications in a form we call Forall-Exist Hoare Tuples (FEHT) (also called \emph{refinement quadruples} \cite{BartheCK13} or \emph{RHLE triples} \cite{DickersonYZD22}), which have the form
\begin{align*}
	\rel{\Phi}{\prog_1 \oast \cdots \oast \prog_k}{\prog_{k+1} \oast \cdots \oast \prog_{k+l}}{\Psi},
\end{align*}
where $\prog_1, \ldots, \prog_{k+l}$ are (possibly identical) programs and $\Phi, \Psi$ are first-order formulas that relate $k+l$ different program runs.
The FEHT is valid if for \emph{all} $k+l$ initial states that satisfy $\Phi$, and for \emph{all} possible executions of $\prog_1, \ldots, \prog_k$ there \emph{exist} executions of $\prog_{k+1}, \ldots, \prog_{k+l}$ such that the final states satisfy $\Psi$.
For example, GNI can be expressed as $\rel{l_1=l_2}{\prog}{\prog}{o_1=o_2}$, where $l_1$ and $o_1$ (resp.~$l_2$ and $o_2$) refer to the value of $l$ and $o$ in the first (resp.~second) program copy.
That is, for \emph{any} two initial states $\sigma_1, \sigma_2$ with identical values for $l$ (but possibly different values for $h$), and \emph{any} final state $\sigma_1'$ reachable by executing $\prog$ from $\sigma_1$, there \emph{exists} some final state $\sigma_2'$ (reachable from $\sigma_2$ by executing $\prog$) that agrees with $\sigma_1'$ in the value of $o$.
The program in \Cref{fig:intro-example} satisfies this FEHT.
In the terminology of Clarkson and Schneider \cite{ClarksonS10}, GNI is a \emph{hyperliveness} property, hence the name of our paper. 
Intuitively, the term hyperliveness stems from the fact that -- due to the existential quantification in FEHTs -- GNI  reasons about the existence of a particular execution. 
Similar to the definition of liveness in temporal properties \cite{AlpernS85}, we can, therefore, satisfy GNI by \emph{adding} sufficiently many execution traces \cite{CoenenFST19}.

\paragraph{Verification Using a Program Logic.}

For  \emph{finite}-state hardware systems, many automated verification methods for hyperliveness properties (e.g., in the form of FEHTs) have been proposed \cite{ClarksonFKMRS14,HsuSB21,BeutnerFFM23,FinkbeinerRS15,BeutnerF23,BeutnerF23b,CoenenFST19}.
In contrast, for \emph{infinite}-state software, the verification of FEHTs is notoriously difficult; FEHTs mix quantification of different types, so we cannot employ purely over-approximate reasoning principles (as is possible for $k$-safety).
Most existing approaches for software verification, therefore, require substantial user interaction, e.g., in the form of a custom Horn-clause template \cite{UnnoTK21}, a user-provided abstraction \cite{BeutnerF22}, or a deductive proof strategy \cite{DickersonYZD22,BartheCK13}.
See \Cref{sec:relatedWork} for more discussion.

In this paper, we put forward an automatic algorithm for the verification of FEHTs.
Our method is rooted in a novel \emph{program logic}, which we call Forall-Exist Hoare Logic (FEHL) (in \Cref{sec:FEHL}).  
Similar to many program logics for $k$-safety properties \cite{SousaD16,ChenFD17}, our logic focuses on one of the programs involved in the verification at any given time (by, e.g., symbolically executing one step in one of the programs) and thus lends itself to automation. 
We show that FEHL is sound and complete (relative to a complete proof system for over- and under-approximate unary Hoare triples).

\paragraph{Automated Verification.}

Our verification algorithm -- presented in \Cref{sec:algorithm} -- then leverages FEHL for the analysis of FEHTs.
During this analysis, the key algorithmic challenge is to find suitable instantiations for nondeterministic choices made in existentially quantified executions. 
Our algorithm avoids a direct instantiation and instead treats the outcome of the nondeterministic choice \emph{symbolically}, allowing an instantiation at a later point in time.
Formally, we define the concept of a \emph{parametric assertion}. 
Instead of capturing a set of states, a parametric assertion defines a function that maps concrete values for a set of parameters (in our case, the nondeterministic choices in existentially quantified programs whose concrete instantiations we have postponed) to sets of states. 
Our algorithm then recursively computes a \emph{parametric postcondition} and delegates the search for appropriate instantiations of the parameters to an SMT solver. 
Crucially, our algorithm only explores a restricted class of program alignments (as guided by FEHL).
Therefore, the resulting constraints are ordinary (first-order) SMT formulas, which can be handled using off-the-shelf SMT solvers.

\paragraph{Implementation and Experiments.}

We implement our algorithm in a tool called \tool{} and compare it with existing approaches for the verification of $\forall^*\exists^*$ properties (in \Cref{sec:implementation}).
As \tool{} can resort to highly optimized off-the-shelf SMT solvers, it outperforms existing approaches (which often rely on custom solving strategies) in many benchmarks.

\section{Preliminaries}\label{sec:prelim}

\paragraph{Programs.}

Let $\progVars$ be a set of program variables. 
We consider a simple (integer-valued) programming language generated by the following grammar.
\begin{align*}
	\prog, \progg := \myskip \mid x \myassign e \mid \myassume(b) \mid \myif(b, \prog, \progg) \mid \mywhile(b, \prog) \mid \prog \mysemi \progg \mid x \myassign \star
\end{align*}
where $x \in \progVars$ is a variable, $e$ is a (deterministic) arithmetic expressions over variables in $\progVars$, and $b$ is a (deterministic) boolean expression.
$\myskip{}$ denotes the program that does nothing; $x \myassign e$ assigns $x$ the result of evaluating $e$; $\myassume(b)$ assumes that $b$ holds, i.e., does not continue execution from states that do not satisfy $b$; $\myif(b, \prog, \progg)$ executes $\prog$ if $b$ holds and otherwise executes $\progg$; $\mywhile(b, \prog)$ executes $\prog$ as long as $b$ holds; $\prog\mysemi \progg$ executes $\prog$ followed by $\progg$; and $x \myassign \star$ assigns $x$ some nondeterministically chosen integer. 
For an arithmetic expression $e$, we write $\mathit{Vars}(e) \subseteq \progVars$ for the set of all variables used in the expression.

We endow our language with a standard operational semantics operating on states $\sigma : \progVars \to \intSet$.
Given a program $\prog$, we write $\sem{\prog}(\sigma, \sigma')$ whenever $\prog$ -- when executed from state $\sigma$ -- \emph{can} terminate in state $\sigma'$.
Our semantics is defined as expected, and we give a full definition in \ifFull{\Cref{app:semantics}}{the full version \cite{full}}.

Given program states $\sigma_1 : \progVars \to \intSet$ and $\sigma_2 : \progVars' \to \intSet$ with $\progVars \cap \progVars' = \emptyset$, we write $\sigma_1 \oplus \sigma_2 : (\progVars \cup \progVars') \to \intSet$ for the combined state, that behaves as $\sigma_1$ on $\progVars$ and as $\sigma_2$ on $\progVars'$.
For $i \in \nat$, we define $\progVars_i := \{x_i \mid x \in \progVars\}$ as a set of indexed program variables. 

\paragraph{Assertions.}

An assertion $\Phi$ is a first-order formula over variables in $\progVars$ (or in the relational setting over $\bigcup_{i=1}^k \progVars_i$ for some $k$).
Given a state $\sigma$, we write $\sigma \models \Phi$ if $\sigma$ satisfies $\Phi$.
We assume that assertions stem from an arbitrarily expressive background theory such that every set of states can be expressed as a formula.
This allows us to sidestep the issue of \emph{expressiveness} in the sense of Cook \cite{Cook78} (see, e.g., \cite{OHearn20,Yang07,SousaD16} for similar treatments).

\paragraph{Hyperliveness Specifications.}

Our verification algorithm targets specifications that combine universal and existential quantification, similar to \emph{RHLE triples} \cite{DickersonYZD22} and \emph{refinement quadruples} \cite{BartheCK13}:

\begin{definition}
	A Forall-Exist Hoare Tuple (FEHT) has the form 
	\begin{align*}
		\rel{\Phi}{\prog_1 \oast \cdots \oast \prog_{k}}{\prog_{k+1} \oast \cdots \oast \prog_{k+l}}{\Psi},
	\end{align*}
	where $\Phi, \Psi$ are assertions over $\bigcup_{i=1}^{k+l} \progVars_i$, and $\prog_1, \ldots, \prog_{k+l}$ are programs over variables $\progVars_1, \ldots, \progVars_{k+l}$, respectively. 
	The FEHT is \emph{valid} if for all states $\sigma_1, \ldots, \sigma_{k+l}$ (with domains $\progVars_1, \ldots, \progVars_{k+l}$, respectively) and  $\sigma'_1, \ldots, \sigma'_{k}$ such that $\bigoplus_{i=1}^{k+l} \sigma_i \models \Phi$ and $\sem{\prog_i}(\sigma_i, \sigma'_i)$ for all $i \in [1,k]$, there exist states $\sigma'_{k+1}, \ldots, \sigma'_{k+l}$ such that $\sem{\prog_i}(\sigma_i, \sigma'_i)$ for all $i \in [k+1,k+l]$ and $\bigoplus_{i=1}^{k+l} \sigma'_i \models \Psi$.
\end{definition}

That is, we quantify universally over initial states for all $k+l$ programs (under the assumption that they, together, satisfy $\Phi$) and also universally over executions of $\prog_1, \ldots, \prog_k$.
Afterward, we quantify \emph{existentially} over executions of $\prog_{k+1}, \ldots, \prog_{k+l}$ and require that the final states of all $k+l$ executions, together, satisfy the postcondition $\Psi$.
A relational property usually refers to $k+l$ executions of the \emph{same} program $\prog$ (operating on variables in $\progVars$); we can model this by using $\alpha$-renamed copies $\rename{\prog}{1}, \ldots, \rename{\prog}{k+l}$ where each $\rename{\prog}{i}$ is obtained from $\prog$ by replacing each variable $x \in \progVars$ with $x_i \in \progVars_i$.
FEHTs capture a range of important properties, including e.g., non-inference \cite{McLean94}, opacity \cite{ZhangYZ19}, GNI \cite{McCullough88}, refinement \cite{Wirth71}, software doping \cite{BiewerDFGHHM22}, and robustness \cite{ChaudhuriGL12}.
It is easy to see that FEHTs can also express (purely universal) $k$-safety properties over programs $\prog_1, \ldots, \prog_k$ as $\rel{\Phi}{\prog_1 \oast \cdots \oast \prog_{k}}{\epsilon}{\Psi}$, where $\epsilon$ denotes the empty sequence of programs.

\section{Forall-Exist Hoare Logic}\label{sec:FEHL}

\newcommand{\greyName}[1]{{\color{black!50}{\texttt{#1}}}}

\begin{figure}[!t]
	\def\ScoreOverhang{2pt}
	\footnotesize
	\vspace{-3mm}

	\begin{minipage}{0.33\linewidth}
		\centering
		\scalebox{0.9}{\parbox{\linewidth}{
				\def\defaultHypSeparation{\hskip .2in}
				\begin{topprooftree}{\labelText{\texttt{($\forall$-Reorder)}}{rule:forall-comm}}
					\AxiomC{$\vdash\rel{\Phi}{\univC_2 \oast \univC_1}{\existsC}{\Psi}$}
					\UnaryInfC{$\vdash\rel{\Phi}{\univC_1 \oast \univC_2}{\existsC}{\Psi}$}
				\end{topprooftree}
		}}
	\end{minipage}%
	\begin{minipage}{0.33\linewidth}
		\centering
		\scalebox{0.9}{\parbox{\linewidth}{
				\def\defaultHypSeparation{\hskip .2in}
				\begin{topprooftree}{\labelText{\texttt{($\forall$-Skip-I)}}{rule:forall-intro}}
					\AxiomC{$\vdash\rel{\Phi}{\prog\mysemi\myskip \oast \univC}{\existsC}{\Psi}$}
					\UnaryInfC{$\vdash\rel{\Phi}{\prog \oast \univC}{\existsC}{\Psi}$}
				\end{topprooftree}
		}}
	\end{minipage}%
	\begin{minipage}{0.33\linewidth}
		\centering
		\scalebox{0.9}{\parbox{\linewidth}{
				\def\defaultHypSeparation{\hskip .2in}
				\begin{topprooftree}{\labelText{\texttt{($\forall$-Skip-E)}}{rule:forall-elim}}
					\AxiomC{$\vdash \rel{\Phi}{\univC}{\existsC}{\Psi}$}
					\UnaryInfC{$\vdash \rel{\Phi}{\myskip \oast \univC}{\existsC}{\Psi}$}
				\end{topprooftree}
		}}
	\end{minipage}

	\begin{minipage}{0.39\linewidth}
		\centering
		\scalebox{0.88}{\parbox{\linewidth}{
				\def\defaultHypSeparation{\hskip .2in}
				\begin{topprooftree}{\labelText{\texttt{($\forall$-If)}}{rule:forall-if}}
					\AxiomC{\stackanchor{$\vdash\rel{\Phi \land b}{\prog_1\mysemi \prog_3 \oast \univC}{\existsC}{\Psi}$}{$\vdash\rel{\Phi \land \neg b}{\prog_2\mysemi \prog_3 \oast \univC}{\existsC}{\Psi}$}}
					\UnaryInfC{$\vdash\rel{\Phi}{\myif(b, \prog_1, \prog_2)\mysemi \prog_3 \oast \univC}{\existsC}{\Psi}$}
				\end{topprooftree}
		}}
	\end{minipage}%
	\begin{minipage}{0.31\linewidth}
		\centering
		\scalebox{0.88}{\parbox{\linewidth}{
				\def\defaultHypSeparation{\hskip .2in}
				\begin{topprooftree}{\labelText{\texttt{($\forall$-Step)}}{rule:forall-step}}
					\AxiomC{\stackanchor{$\vdash\tripH{\Phi}{\prog_1}{\Phi'}$}{$\vdash\rel{\Phi'}{\prog_2 \oast \univC}{\existsC}{\Psi}$}}
					\UnaryInfC{$\vdash\rel{\Phi}{\prog_1 \mysemi \prog_2 \oast \univC}{\existsC}{\Psi}$}
				\end{topprooftree}
		}}
	\end{minipage}
		\begin{minipage}{0.29\linewidth}
		\centering
		\scalebox{0.88}{\parbox{\linewidth}{
				\def\defaultHypSeparation{\hskip .2in}
				\begin{topprooftree}{\labelText{\texttt{($\exists$-Step)}}{rule:exists-step}}
					\AxiomC{\stackanchor{$\vdash\tripN{\Phi}{\prog_1}{\Phi'}$}{$\vdash\rel{\Phi'}{\univC}{\prog_2 \oast \existsC}{\Psi}$}}
					\UnaryInfC{$\vdash\rel{\Phi}{\univC}{\prog_1\mysemi\prog_2 \oast \existsC}{\Psi}$}
				\end{topprooftree}
		}}
	\end{minipage}%
	
	\vspace{-2mm}

	\begin{minipage}{0.2\linewidth}
		\centering
		\scalebox{0.9}{\parbox{\linewidth}{
				\begin{topprooftree}{\labelText{\texttt{(Done)}}{rule:done}}
					\AxiomC{$\vdash \rel{\Phi}{\epsilon}{\epsilon}{\Phi}$}
				\end{topprooftree}
		}}
	\end{minipage}
	\begin{minipage}{0.4\linewidth}
		\vspace{3mm}
		\centering
		\scalebox{0.9}{\parbox{\linewidth}{
				\def\defaultHypSeparation{\hskip .2in}
				\begin{topprooftree}{\labelText{\texttt{($\forall$-Assume)}}{rule:forall-assume}}
					\AxiomC{$\vdash\rel{\Phi \land b}{\prog \oast \univC}{\existsC}{\Psi}$}
					\UnaryInfC{$\vdash\rel{\Phi}{\myassume(b) \mysemi \prog \oast \univC}{\existsC}{\Psi}$}
				\end{topprooftree}
		}}
	\end{minipage}%
	\begin{minipage}{0.4\linewidth}
		\centering
		\scalebox{0.9}{\parbox{\linewidth}{
				\vspace{3.3mm}
				\def\defaultHypSeparation{\hskip .2in}
				\begin{topprooftree}{\labelText{\texttt{($\exists$-Assume)}}{rule:exists-assume}}
					\AxiomC{$\Phi \Rightarrow b$}
					\AxiomC{$\vdash\rel{\Phi}{\univC}{\prog \oast \existsC}{\Psi}$}
					\BinaryInfC{$\vdash\rel{\Phi}{\univC}{\myassume(b) \mysemi \prog \oast \existsC}{\Psi}$}
				\end{topprooftree}
		}}
	\end{minipage}%
	
	\vspace{-2mm}
	
	\begin{minipage}{0.37\linewidth}
		\vspace{3mm}
		\centering
		\scalebox{0.9}{\parbox{\linewidth}{
				\def\defaultHypSeparation{\hskip .05in}
				\begin{topprooftree}{\labelText{\texttt{($\forall$-Choice)}}{rule:forall-inf-nd}}
					\AxiomC{$\vdash\rel{\exists x\ldot \Phi}{\prog \oast \univC}{\existsC}{\Psi}$}
					\UnaryInfC{$\vdash\rel{\Phi}{x \myassign \star\mysemi \prog \oast \univC}{\existsC}{\Psi}$}
				\end{topprooftree}
		}}
	\end{minipage}
	\begin{minipage}{0.63\linewidth}
		\centering
		\scalebox{0.9}{\parbox{\linewidth}{
				\vspace{3.5mm}
				\def\defaultHypSeparation{\hskip .2in}
				\begin{topprooftree}{\labelText{\texttt{($\exists$-Choice)}}{rule:exists-inf-nd}}
					\AxiomC{$x \not\in \mathit{Vars}(e)$}
					\AxiomC{$\vdash\rel{(\exists x\ldot \Phi) \land x = e}{\univC}{\prog \oast \existsC}{\Psi}$}
					\BinaryInfC{$\vdash\rel{\Phi}{\univC}{x \myassign \star\mysemi \prog  \oast \existsC}{\Psi}$}
				\end{topprooftree}
		}}
	\end{minipage}

	\caption{Selection of core proof rules of FEHL}\label{fig:core_rules}
\end{figure}

The verification steps of our constraint-based algorithm (presented in \Cref{sec:algorithm}) are guided by the proof rules of a novel program logic operating on FEHTs, which we call Forall-Exist Hoare Logic (FEHL).

\subsection{Core Rules}

We depict a selection of core rules in \Cref{fig:core_rules}; a full overview can be found in \ifFull{\Cref{app:fehl}}{\cite{full}}.
We write $\univC$ (resp.~$\existsC$) to abbreviate a list $\prog_1 \oast \cdots \oast \prog_k$ of programs that are universally (resp.~existentially) quantified.
Rule \nameref{rule:forall-comm} allows for the reordering of universally quantified programs; \nameref{rule:forall-intro} rewrites a program $\prog$ into $\prog \mysemi \myskip$; \nameref{rule:forall-elim} removes a single $\myskip$-instruction; and \nameref{rule:done} derives a FEHL with an empty program sequence. 
Using $\myskip$-insertions and reordering (and the analogous rules for existentially quantified programs), we can always bring a program in the form $\prog_1 \mysemi \prog_2$, targeted by the remaining rules. 
Rule \nameref{rule:forall-if} embeds the branching condition of a conditional into the preconditions of both branches. 
Rules \nameref{rule:forall-step} and \nameref{rule:exists-step} allow us to resort to unary reasoning over parts of the program.
These rules make the multiplicity of techniques  developed for unary reasoning (e.g., symbolic execution \cite{King76} and predicate transformers \cite{0067387}) applicable to the verification of hyperproperties in the form of FEHTs.
For universally quantified programs of the form $\prog_1 \mysemi \prog_2$,  \nameref{rule:forall-step} requires an auxiliary assertion $\Phi'$ that should hold after \emph{all} executions of $\prog_1$ from $\Phi$.  
We can express this using the standard (non-relational) Hoare triple (HT) $\tripH{\Phi}{\prog_1}{\Phi'}$ \cite{Hoare69}.
The second premise then ensures that the remaining FEHT (after $\prog_1$ has been executed) holds.  
For existentially quantified programs, we, instead, employ an underapproximation. 
In \nameref{rule:exists-step}, we, again, execute $\prog_1$ but use an \emph{Under-Approximate Hoare triple} (UHT) $\tripN{\Phi}{\prog_1}{\Phi'}$.
The UHT $\tripN{\Phi}{\prog_1}{\Phi'}$ holds if for all states $\sigma$ with $\sigma \models \Phi$, there \emph{exists} a state $\sigma'$ such that $\sem{\prog_1}(\sigma, \sigma')$ and $\sigma' \models \Phi'$.

\begin{remark}
	UHTs behave similar to \emph{Incorrectness Triples} (ITs) \cite{OHearn20,VriesK11} in that they reason about the existence of a particular set of executions. 
	The key difference is that ITs reason backward (all states in $\Phi'$ are reachable from some state in $\Phi$), whereas UHTs reason in a forward direction (all states in $\Phi$ can reach $\Phi'$).
	See, e.g., \emph{Lisbon Triples} \cite[\S 5]{MollerOH21} and \emph{Outcome Triples} \cite{ZilbersteinDS23} for related approaches.
	We will later show that FEHL is complete when equipped with some complete proof system for UHTs (cf.~\Cref{theo:completeness}). 
	In \ifFull{\Cref{app:ehl}}{the full version \cite{full}}, we show that there exists at least one complete proof system for UHTs.
	 \demo
\end{remark}

For $\myassume$ statements, \nameref{rule:forall-assume} strengthens the precondition by the assumed expression $b$; any state that does not satisfy $b$ causes a (universally quantified) execution to halt and renders the FEHT vacuously valid. 
In contrast, \nameref{rule:exists-assume} assumes that all states in $\Phi$ satisfy $b$; if any state in $\Phi$ does not satisfy $b$, the FEHT is invalid.
Likewise, the handling of a nondeterministic assignment $x \myassign \star$  differs based on whether we consider a universally quantified or existentially quantified program. 
In the former case, \nameref{rule:forall-inf-nd} removes all knowledge about the value of $x$ within the precondition by quantifying $x$ existentially (thus enlarging the precondition).
In the latter (existentially quantified) case, we can, in a forward-style execution, choose \emph{any} concrete value for $x$.
\nameref{rule:exists-inf-nd} formalizes this intuition: we first invalidate all knowledge about $x$ and then assert that $x = e$ for some arbitrary expression $e$ that does not depend on $x$. 
In our automated analysis (cf.~\Cref{sec:algorithm}), we use \nameref{rule:exists-inf-nd}, but -- instead of fixing some concrete value (or expression) at application time -- we postpone the concrete instantiation by treating the value \emph{symbolically}.

\subsection{Asynchronous Loop Reasoning}\label{sec:loops}

\begin{figure}[!t]
	\def\ScoreOverhang{2pt}
	\centering
	\centering
	\scalebox{0.9}{\parbox{\linewidth}{
			\begin{topprooftree}{\labelText{\texttt{(Loop-Counting)}}{rule:loop-count}}
				\def\defaultHypSeparation{\hskip .15in}
				\AxiomC{\stackanchor{\stackanchor{\stackanchor{$k \geq 1$, $B \geq 1$}{$c_1, \ldots, c_{k+l} \in [1, B]$}}{$\mathbb{I}_1, \ldots, \mathbb{I}_{B+1}$}}{\stackanchor{$\Phi \Rightarrow \mathbb{I}$}{\stackanchor{$\mathbb{I} \Rightarrow \bigwedge_{i=2}^{k+l} (b_1 \leftrightarrow b_i)$}{$\mathbb{I} = \mathbb{I}_1 = \mathbb{I}_{B+1}$}}}}
				\AxiomC{\stackanchor
					{$\left[ \vdash\relBig{\mathbb{I}_j \land \bigwedge\limits_{\substack{i = 1\mid c_i \geq j}}^{k+l} \!\!\!b_i}{\bigoast\limits_{\substack{i=1 \mid c_i \geq j}}^{k}  \prog_i\;}{\bigoast\limits_{\substack{i=k+1 \mid c_i \geq j}}^{k+l} \prog_i}{\mathbb{I}_{j+1} \land  \bigwedge\limits_{\substack{i = 1\mid c_i > j}}^{k+l} \!\!\! b_i}\right]_{j=1}^{B}$}
					{$\vdash\relBig{\mathbb{I} \land \bigwedge\limits_{i=1}^{k+l} \neg b_i}{\bigoast\limits_{i=1}^{k} \progg_i \oast \univC\;}{\bigoast\limits_{i=k+1}^{k+l}\progg_i \oast \existsC}{\Psi}$}
				}
				\BinaryInfC{$\vdash \rel{\Phi}{ \bigoast\limits_{i=1}^k \mywhile(b_i, \prog_i) \mysemi \progg_i \oast \univC }{\bigoast\limits_{i=k+1}^{k+l} \mywhile(b_i, \prog_i) \mysemi \progg_i \oast \existsC}{\Psi}$}
			\end{topprooftree}
	}}
	
	\vspace{-3mm}
	
	\caption{Counting-based loop rule for FEHL }\label{fig:loop_rule_count}
\end{figure}

A particular challenge when reasoning about relational properties is the alignment of loops.
In FEHL, we propose a novel counting-based loop rule that supports asynchronous alignments while still admitting good automation.
Consider the rule \nameref{rule:loop-count} (in \Cref{fig:loop_rule_count}), which assumes $k \geq 1$ universally and $l$ existentially quantified loops. 
The rule requires a loop invariant $\mathbb{I}$ that \textbf{(1)} is implied by the precondition ($\Phi \Rightarrow \mathbb{I}$), \textbf{(2)} ensures simultaneous termination of all loops ($\mathbb{I} \Rightarrow \bigwedge_{i=2}^{k+l} (b_1 \leftrightarrow b_i)$), and \textbf{(3)} is strong enough to establish the postcondition for the program suffixes $\progg_1, \ldots, \progg_{k+l}$ executed after the loops.
The key difference from a simple synchronous traversal is that, in each ``iteration'', we execute the bodies of the loops for possibly different numbers of times. 
Concretely, \nameref{rule:loop-count} asks for natural numbers $c_1, \ldots, c_{k+l}$ (ranging between $1$ and some arbitrary upper bound $B$), and -- starting from the invariant $\mathbb{I}$ -- we execute each $\prog_i$ $c_i$ times.
Crucially, we need to make sure that each $\prog_i$ will execute \emph{at least} $c_i$ times, i.e., the guard $b_i$ holds after each of the first $c_i-1$ executions.
In particular, we cannot na\"ively analyze $c_i$ copies of $\prog_i$ composed via $\mysemi$ as this might introduce additional executions of $\prog_i$ that would not happen in $\mywhile(b_i, \prog_i)$.
To ensure this, \nameref{rule:loop-count} demands $B+1$ intermediate assertions $\mathbb{I}_1, \ldots, \mathbb{I}_{B+1}$.
In the $j$th iteration (for $1 \leq j \leq B$), we (symbolically) execute -- from $\mathbb{I}_j$ -- all loop bodies $\prog_i$ that we want to execute at least $j$ times (i.e., all loop bodies $\prog_i$ where $c_i \geq j$).
We require that \textbf{(1)} the postcondition $\mathbb{I}_{j+1}$ is derivable, and \textbf{(2)} the guards of all loops that we want to execute \emph{more} than $j$ times (i.e., loops where $c_i > j$) evaluate to true.

\begin{figure}[!t]

\newsavebox{\myboxi}
\newsavebox{\myboxii}
\begin{lrbox}{\myboxi}
\begin{exampleCode}
$y_1$ = $y_1$ - $1$$\mysemi$
$x_1$ = $4$ * $x_1$
\end{exampleCode}
\end{lrbox}
\begin{lrbox}{\myboxii}
\begin{exampleCode}
$z_2$ = $\star$$\mysemi$ 
$y_2$ = $y_2$ - $z_2$$\mysemi$ 
$x_2$ = $2$ * $x_2$
\end{exampleCode}
\end{lrbox}
	
\begin{minipage}[b]{0.35\linewidth}
\begin{lrbox}{\mybox}
\begin{exampleCode}
$y_1$ = $x_1$$\mysemi$
while ($y_1$ > $0$)
@@$y_1$ = $y_1$ - $1$$\mysemi$
@@$x_1$ = $4$ * $x_1$
\end{exampleCode}
\end{lrbox}

\begin{align*}
	\prog_1 := \begin{cases}
		\scalebox{0.95}{\text{\usebox{\mybox}}}
	\end{cases}
\end{align*}

\begin{lrbox}{\mybox}
\begin{exampleCode}
$y_2$ = $2$ * $x_2$$\mysemi$
while ($y_2$ > $0$)
@@$z_2$ = $\star$$\mysemi$ 
@@$y_2$ = $y_2$ - $z_2$$\mysemi$ 
@@$x_2$ = $2$ * $x_2$
\end{exampleCode}
\end{lrbox}

\begin{align*}
	\prog_2 := \begin{cases}
		\scalebox{0.95}{\text{\usebox{\mybox}}}
	\end{cases}
\end{align*}

\subcaption{}\label{fig:asynchronous-loop}

\end{minipage}%
\begin{minipage}[b]{0.65\linewidth}
	
	\begin{subfigure}{\linewidth}
		\centering
		\scalebox{0.95}{\parbox{\linewidth}{
				\begin{align*}
					\relBigg{\begin{matrix}
							\mathbb{I}_1 \, \land \\
							y_1 > 0 \, \land \\
							y_2 > 0
					\end{matrix}}{\begin{matrix}
							\scalebox{1}{\usebox{\myboxi}}
						\end{matrix}\;}{\begin{matrix}
							\scalebox{1}{\usebox{\myboxii}}
						\end{matrix}\;}{\begin{matrix}
							\mathbb{I}_2\, \land\\ y_2 > 0
					\end{matrix}}
				\end{align*}
		}}
		\subcaption{}\label{fig:step1}
	\end{subfigure}\\
	\begin{subfigure}{\linewidth}
		\centering
		\scalebox{0.95}{\parbox{\linewidth}{
				\begin{align*}
					\relBigg{\begin{matrix}
							\mathbb{I}_2 \, \land\\ y_2 > 0
					\end{matrix}}{\,\epsilon\,}{\begin{matrix}
							\scalebox{1}{\usebox{\myboxii}}
						\end{matrix}\;}{\begin{matrix}
							\mathbb{I}_3
					\end{matrix}}
				\end{align*}
		}}
		\subcaption{}\label{fig:step2}
	\end{subfigure}

\end{minipage}

	\caption{In \Cref{fig:asynchronous-loop}, we depict two example programs. In \Cref{fig:step1,fig:step2}, we give two intermediate FEHT verification obligations (cf.~\Cref{ex:counting-example}).}
\end{figure}

\begin{example}\label{ex:counting-example}
	Consider the two example programs $\prog_1, \prog_2$ in \Cref{fig:asynchronous-loop} and the FEHT $\rel{x_1=x_2}{\prog_1}{\prog_2}{x_1=x_2}$.
	To see that this FEHT is valid, we can, in each loop iteration, always choose $z_2 = 1$.
	In this case, $\prog_1$ quadruples the value of $x_1$ for $x_1$ times and $\prog_2$ doubles the value of $x_2$ for $2 x_2$ times, which, assuming $x_1 = x_2$, computes the same result ($x_1 = x_2 \rightarrow 4^{x_1} x_1 = 2^{2x_2}x_2$).
	Verifying this example automatically is challenging as both loops are executed a different number of times, so we cannot align the loops in lockstep.
	Likewise, computing independent (unary) summaries of both loops requires complex non-linear reasoning.
	Instead, \nameref{rule:loop-count} enables an asynchronous alignment:	
	After applying \nameref{rule:forall-step} and \nameref{rule:exists-step}, we are left with precondition $x_1 = x_2 \land y_2 = 2y_1$. 
	We use \nameref{rule:loop-count} and align the loops such that every loop iteration in $\prog_1$ is matched by \emph{two} iterations in $\prog_2$, which allows us to use a simple (linear) invariant.
	We set $c_1 := 1, c_2 := 2$ and define $\mathbb{I} := x_1 = x_2 \land y_2 = 2y_1$, $\mathbb{I}_1 := \mathbb{I}_3 := \mathbb{I}$, and $\mathbb{I}_2 := x_1 = 2x_2 \land y_2 = 2y_1 + 1$. 
	Note that $\mathbb{I}$ implies the desired postcondition ($x_1 = x_2$).
	To establish that $\mathbb{I}$ serves as an invariant, we need to discharge the two proof obligations depicted in \Cref{fig:step1,fig:step2}.
	The obligation in \Cref{fig:step1} (corresponding to iteration $j=1$) establishes that \textbf{(1)} $\mathbb{I}_2$ is a provable postcondition after executing both loop bodies from $\mathbb{I}_1$ and \textbf{(2)} that the loop in $\prog_2$ will execute at least one more time, i.e., $y_2 > 0$.
	We can easily discharge this FEHT using \nameref{rule:forall-step}, \nameref{rule:exists-step}, and \nameref{rule:exists-inf-nd} by choosing $z_2$ to be $1$ (note that if $y_2 = 2y_1$ and $y_2 > 0$, then $y_2 - 1 > 0$).
	The obligation in \Cref{fig:step2} corresponds to iteration $j = 2$, where we only execute the body of $\prog_2$.
	We can, again, easily discharge this FEHT using \nameref{rule:exists-step} and \nameref{rule:exists-inf-nd} (again, choosing $z_2$ to be $1$).
	\demo
\end{example}

\subsection{Soundness and Completeness}

We can show that our proof system is sound and complete:

\begin{restatable}[Soundness]{theorem}{soundness}\label{theo:sound}
	Assume that $\vdash \tripH{\cdot}{\cdot}{\cdot}$ and $\vdash \tripN{\cdot}{\cdot}{\cdot}$ are sound proof systems for HTs and UHTs, respectively. 
	If $\vdash \rel{\Phi}{\univC}{\existsC}{\Psi}$ then $\rel{\Phi}{\univC}{\existsC}{\Psi}$ is valid.
\end{restatable}

\begin{restatable}[Completeness]{theorem}{completeness}\label{theo:completeness}
	Assume that $\vdash \tripH{\cdot}{\cdot}{\cdot}$ and $\vdash \tripN{\cdot}{\cdot}{\cdot}$ are complete proof systems for HTs and UHTs, respectively. 
	If $\rel{\Phi}{\univC}{\existsC}{\Psi}$ is valid then $\vdash \rel{\Phi}{\univC}{\existsC}{\Psi}$.
\end{restatable}

Completeness follows easily by making extensive use of \emph{unary} reasoning via (U)HTs, similar to the completeness-proof of relational Hoare logic for $k$-safety properties \cite{NagasamudramN21}.
In fact, \nameref{rule:forall-step}, \nameref{rule:exists-step}, \nameref{rule:done} along with the reordering rules \nameref{rule:forall-comm}, \nameref{rule:forall-intro}, and \nameref{rule:forall-elim} (and their analogous counterparts for existentially quantified programs) already suffice for completeness (see \ifFull{\Cref{app:fehl_compltness}}{\cite{full}}). 
In the following, we leverage the \emph{soundness} of FEHL's rules to guide our automated verification.

\section{Automated Verification of Hyperliveness}\label{sec:algorithm}

Our automated verification algorithm for FEHTs follows a \emph{strongest postcondition} computation, as is widely used in the verification of non-relational properties \cite{AhrendtBBBGHMMRSS05,HenzingerJMM04,PasareanuV04} and $k$-safety properties \cite{SousaD16,ChenFD17}.
However, due to the inherent presence of \emph{existential} quantification in FEHT, the strongest postcondition does, in general, not exist.
For example, both $\rel{\top}{\epsilon}{x \myassign \star}{x = 1}$ and $\rel{\top}{\epsilon}{x \myassign \star}{x = 2}$ are valid but $\rel{\top}{\epsilon}{x \myassign \star}{x = 1 \land x = 2 \equiv \bot}$ is clearly not.
Instead, our algorithm uses the proof rules of FEHL and treats the concrete value for nondeterministic choices in existentially quantified executions symbolically.
I.e., we view the outcome as a fresh variable (called a \emph{parameter}) that can be instantiated later. 
This idea of instating nondeterminism at a later point in time has already found successful application in many areas, such as existential variables in Coq or symbolic execution \cite{King76}.
Our analysis brings these techniques to the realm of  hyperproperty verification, which we show to yield an effective automated verification algorithm. 
In the following, we formally introduce parametric assertions and postconditions (in \Cref{sec:para_assertions}) and show how we can compute them using the rules of FEHL (in \Cref{sec:algorithmDetails,sec:algorithmDetailsLoop}).

\subsection{Parametric Assertions and Postconditions}\label{sec:para_assertions}

We assume that $\parameter = \{\mu_1, \ldots, \mu_n\}$ is a set of \emph{parameters}. 
In FEHTs, we use assertions (formulas) over $\bigcup_{i=1}^{k+l} \progVars_i$, which we interpret as sets of (relational) states.  
A parametric assertion generalizes this by viewing an assertion as a function mapping \emph{into} sets of (relational) states.
Formally, a \emph{parametric assertion} is a pair $(\postF, \condF)$ where $\postF$ is a formula over $\bigcup_{i=1}^{k+l} \progVars_i \cup \parameter$ (called the \emph{function-formula}), and $\condF$ is a formula over $\parameter$ (called the \emph{restriction-formula}).

Given a function-formula $\postF$ (over $\bigcup_{i=1}^{k+l} \progVars_i \cup \parameter$) and a \emph{parameter evaluation} $\kappa : \parameter \to \intSet$, we define $\postF[\kappa]$ as the formula over $\bigcup_{i=1}^{k+l} \progVars_i$ where we fix concrete values for all parameters based on $\kappa$.
We can thus view $\postF$ as a function mapping each parameter evaluation $\kappa$ to the set of states encoded by $\postF[\kappa]$.
During our (forward style) analysis, we will use parameters to postpone nondeterministic choices in existentially quantified programs.
Intuitively, for every parameter evaluation $\kappa$ (i.e., any retrospective choice of the nondeterministic outcome), $\postF[\kappa]$ should describe the reachable states (i.e., strongest postcondition) under those specific outcomes.
However, not all concrete values for the parameters are valid in the sense that they correspond to nondeterministic outcomes that result in actual executions. 
To mitigate this, a parametric assertion $(\postF, \condF)$ includes a restriction-formula $\condF$ (over $\parameter$) which \emph{restrict the domain} of the function encoded by $\postF$, i.e., we only consider those parameter evaluations that satisfy $\condF$.

\begin{example}\label{ex:para-assertion}
	Before proceeding with a formal development, let us discuss parametric assertions informally using an example.  
	Let $\prog_1 := x \myassign \star \mysemi\,\myassume(x \geq 9)$ and $\prog_2 := y \myassign\star\mysemi\,\myassume(y \geq 2)$ and assume we want to prove the FEHT $\rel{\top}{\prog_1}{\prog_2}{x = y}$.
	To verify this tuple in a principled way, we are interested in potential postconditions $\Psi$, i.e., assertions $\Psi$ such that $\rel{\top}{\prog_1}{\prog_2}{\Psi}$ is valid. 
	For example, both $\Psi_1 = x \geq 9 \land y = 2$ and $\Psi_2 = x \geq 9 \land y = 3$ are valid postconditions, but -- as already seen before -- there does not exist a strongest assertion.
	Instead, we capture \emph{multiple} postconditions using the parametric assertion $(\postF, \condF)$ where $\postF := x \geq 9 \land y = \mu$ and $\condF := \mu \geq 2$ for some fresh parameter $\mu \in \parameter$; we say $(\postF, \condF)$ is a \emph{parametric postcondition} for $(\top, \prog_1, \prog_2)$ (cf.~\Cref{def:paraPost}).
	Intuitively, we have used the parameter $\mu$ instead of assigning some fixed integer to $y$.
	For every concrete parameter evaluation $\kappa : \{\mu\} \to \intSet$ such that $\kappa \models \condF$, formula $\postF[\kappa]$ defines the reachable states when using $\kappa(\mu)$ for the choice of $y$.
	Observe how formula $\condF = \mu \geq 2$ restricts the possible set of parameter values, i.e., we may only choose a value for $y$ such that $\myassume(y \geq 2)$ holds.
	\demo
\end{example}

\begin{definition}\label{def:paraPost}
	A \emph{parametric postcondition for $(\Phi, \prog_1, \ldots, \prog_{k+l})$} is a parametric assertion $(\postF, \condF)$ with the following conditions.
	For all states $\sigma_1, \ldots, \sigma_{k+l}$, and $\sigma'_1, \ldots, \sigma'_{k}$ such that $\bigoplus_{i=1}^{k+l} \sigma_i \models \Phi$ and $\sem{\prog_i}(\sigma_i, \sigma'_i)$ for all $i \in [1,k]$ and any parameter evaluation $\kappa$ such that $\kappa \models \condF$ the following holds:
	\textbf{(1)} There exist states $\sigma'_{k+1}, \ldots, \sigma'_{k+l}$ such that $\bigoplus_{i=1}^{k+l} \sigma'_i \models \postF[\kappa]$, and 
	\textbf{(2)} For every $\sigma'_{k+1}, \ldots, \sigma'_{k+l}$ such that $\bigoplus_{i=1}^{k+l} \sigma'_i  \models \postF[\kappa]$ we have $\sem{\prog_i}(\sigma_i, \sigma'_i)$ for all $i \in [k+1, k+l]$.
\end{definition}

Condition \textbf{(1)} captures that no parameter evaluation may restrict universally quantified executions, i.e., if we fix any parameter evaluation $\kappa$ and reachable final states for the universally quantified programs, $\postF[\kappa]$ remains satisfiable.
This effectively states that $\postF[\kappa]$ \emph{over-approximates} the set of executions of universally quantified programs.
Condition \textbf{(2)} requires that all executions of existentially quantified programs allowed under a particular parameter evaluation are also valid executions, i.e., for any fixed parameter evaluation $\kappa$, $\postF[\kappa]$ \emph{under-approximates} the set of executions of the existentially quantified programs.

We can use parametric postconditions to prove FEHTs:

\begin{restatable}{theorem}{paraPostSound}\label{prop:postImplication}
	Let $(\postF, \condF)$ be a parametric postcondition for $(\Phi, \prog_1, \ldots, \prog_{k+l})$. If 
	\begin{align*}
		\textstyle
		\bigforall_{x \in \progVars_1 \cup \cdots \cup \progVars_k} x\ldot \bigexists_{\mu \in \parameter} \mu\ldot \;\; \condF \;\; \land  \;\; \bigforall_{x \in \progVars_{k+1} \cup \cdots \cup \progVars_{k+l}} x \ldot (\postF \Rightarrow \Psi)
	\end{align*}
	holds, then the FEHT $\rel{\Phi}{\prog_1 \oast \cdots \oast \prog_k}{\prog_{k+1} \oast \cdots \oast \prog_{k+l}}{\Psi}$ is valid.
\end{restatable}

Here, we universally quantify over final states in $\prog_1, \ldots, \prog_k$ and existentially quantify over parameter evaluations that satisfy $\condF$ (recall that $\condF$ only refers to $\parameter$). 
The choice of the parameters can thus depend on the final states of universally quantified programs (as in the semantics of FEHTs). 
Afterward, we quantify (again universally) over final states of $\prog_{k+1}, \ldots, \prog_{k+l}$ and state that if $\postF$ holds, so does the postcondition $\Psi$.

\begin{example}\label{ex:para_post}
	Consider the FEHT and parametric postcondition from \Cref{ex:para-assertion}.
	Following \Cref{prop:postImplication}, we construct the SMT formula $\forall x\ldot \exists \mu\ldot  \mu \geq 2 \land \forall y\ldot  \big((x \geq 9 \land y = \mu) \Rightarrow x = y\big)$.
	This formula holds; the FEHT is valid. \demo
\end{example}

Note that $(\postF, \bot)$ is always a parametric postcondition: no parameter evaluation satisfies $\bot$, so the conditions in \Cref{def:paraPost} are vacuously satisfied. 
However, $(\postF, \bot)$ is useless when it comes to proving FEHTs via \Cref{prop:postImplication}.

\subsection{Generating Parametric Postconditions}\label{sec:algorithmDetails}

\begin{algorithm}[!t]
	\caption{Parametric postcondition generation for FEHT verification 
	}\label{algo:verify}

	\vspace{-2mm}
	\begin{minipage}[t]{0.57\textwidth}
\begin{code}
def genpp($\Phi$,$\univC$,$\existsC$):
@if $\univC= \existsC = \epsilon$: 
@@return ($\Phi$,$\top$) //(*\color{comment-color}\nameref{rule:done}*) (*\label{line:done}*)
@else if $\forall \prog \in \univC \cup \existsC \ldot\prog = \mywhile(\_, \_)\mysemi\_$:
@@return genppLoops($\Phi$, $\univC$, $\existsC$) (*\label{line:loops}*)
@else if $\exists \prog \in \univC\ldot \prog \neq \mywhile(\_, \_)\mysemi\_$: 
@@// Take a step in (*\color{comment-color}$\univC$*) (*\label{line:start-univ}*)
@@match $\univC$: 
@@| $\myskip \oast \univC'$: //(*\color{comment-color} \nameref{rule:forall-elim}*) (*\label{line:univ-skip1}*) 
@@@@return genpp($\Phi$,$\univC'$,$\existsC$)
@@| $\myskip\mysemi \prog \oast \univC'$: //(*\color{comment-color}\nameref{rule:forall-step}*)(*\label{line:univ-skip2}*) 
@@@@return genpp($\Phi$,$\prog \oast \univC'$,$\existsC$)
@@| $(\prog_1 \mysemi \prog_2) \mysemi \prog_3 \oast \univC'$: (*\label{line:univ-assoc}*)
@@@@return genpp($\Phi$,$\prog_1 \mysemi (\prog_2 \mysemi \prog_3)\oast \univC'$,$\existsC$)
@@| $\prog \oast \univC'$ when $\prog \neq \_ \mysemi \_$: //(*\color{comment-color} \nameref{rule:forall-intro}*) (*\label{line:univ-skip-intro}*)
@@@@return genpp($\Phi$,$\prog\mysemi\myskip \oast \univC'$,$\existsC$)
@@| $x \myassign e \mysemi \prog \oast \univC'$: //(*\color{comment-color}\nameref{rule:forall-step}*)(*\label{line:univ-assign}*)
@@@@$\Phi'$ := $\exists x'. \Phi[x'/x] \land x = e[x'/x]$ 
@@@@return genpp ($\Phi'$,$\prog \oast \univC'$,$\existsC$)
@@| $\myif(b, \prog_1, \prog_2)\mysemi \prog_3 \oast \univC'$: (*\label{line:univ-if}*)
@@@@//(*\color{comment-color}\nameref{rule:forall-if}*)
@@@@($\postF_1$,$\condF_1$) := 
@@@@@@genpp($\Phi \land b$,$\prog_1\mysemi \prog_3 \oast \univC'$,$\existsC$) 
@@@@($\postF_2$,$\condF_2$) := 
@@@@@@genpp($\Phi \land \neg b$,$\prog_2\mysemi \prog_3 \oast \univC'$,$\existsC$)  
@@@@return ($\postF_1 \lor \postF_2$,$\condF_1 \land \condF_2$)
@@| $\myassume(b)\mysemi \prog \oast \univC'$: //(*\color{comment-color}\nameref{rule:forall-assume}*)(*\label{line:univ-assume}*)
@@@@return genpp($\Phi \land b$,$\prog \oast \univC'$,$\existsC$)
\end{code}
	\end{minipage}%
	\begin{minipage}[t]{0.43\textwidth}
\begin{code}[firstnumber=29]
@| $x \myassign \star\mysemi \prog \oast \univC'$: //(*\color{comment-color}\nameref{rule:forall-inf-nd}*)(*\label{line:univ-nd}*)
@@@$\Phi'$ := $\exists x\ldot \Phi$
@@@return genpp($\Phi'$,$\prog \oast \univC'$,$\existsC$)
@| $\prog \oast \univC'$: //(*\color{comment-color} \nameref{rule:forall-comm}*)
@@@return genpp($\Phi$,$\univC' \oast \prog$,$\existsC$) (*\label{line:univ-reorder}*)  (*\label{line:end-univ}*)
else:
@// Take a step in (*\color{comment-color}$\existsC$*) (*\label{line:start-exists}*)
@match $\existsC$:
@| $\myskip \oast \existsC'$ | $\myskip\mysemi \prog \oast \existsC'$(*\label{line:same-cases-start}*)
@| $(\prog_1 \mysemi \prog_2) \mysemi \prog_3 \oast \existsC'$ 
@| $\prog \oast \existsC'$ when $\prog \neq \_ \mysemi\_$
@| $x \myassign e \mysemi \prog \oast \existsC'$ 
@| $\myif(b, \prog_1, \prog_2)\mysemi \prog_3 \oast \existsC'$:
@@@//As in lines (*\ref{line:univ-skip1}*), (*\ref{line:univ-skip2}*), (*\ref{line:univ-assign}*)
@@@//(*\ref{line:univ-if}*), (*\ref{line:univ-assoc}*), and (*\ref{line:univ-skip-intro}*)(*\label{line:same-cases-done}*)
@| $\myassume(b)\mysemi \prog \oast \existsC'$:
@@@//(*\color{comment-color}\nameref{rule:exists-assume}*)(*\label{line:exists-assume}*)
@@@$\condF_\mathit{assume}$ := 
@@@@@$\bigforall_{x \in \progVars_1 \cup \cdots \cup \progVars_{k+l} } x \ldot  (\Phi \Rightarrow b)$
@@@($\postF$,$\condF$) := 
@@@@@genpp($\Phi \land b$,$\univC$,$\prog \oast \existsC'$)
@@@return ($\postF$,$\condF \land \condF_\mathit{assume}$)
@| $x \myassign \star\mysemi \prog \oast \existsC'$: //(*\color{comment-color}\nameref{rule:exists-inf-nd} *)(*\label{line:exists-nd}*)
@@@$\mu$ := freshParameter()
@@@$\Phi'$ := $(\exists x. \Phi) \land x = \mu$
@@@return genpp($\Phi'$,$\univC$,$\prog \oast  \existsC'$)
@| $\prog \oast \existsC'$: 
@@@return genpp($\Phi$,$\univC$,$\existsC' \oast \prog$)   (*\label{line:end-exists}*) (*\label{line:exists-reorder}*)
\end{code}
	\end{minipage}
	\vspace{-2mm}
\end{algorithm}

\Cref{algo:verify} computes a parametric postcondition based on the proof rules of FEHL from \Cref{sec:FEHL}.
As input, \Cref{algo:verify} expects a formula $\Phi$ over $\bigcup_{i=1}^{k+l}\progVars_i \cup \parameter$ -- think of $\Phi$ as a precondition already containing some parameters -- and two program lists $\univC$ and $\existsC$.
It outputs a parametric postcondition.

\begin{remark}\label{fnlabel}
	For intuition, it is oftentimes helpful to consider $\Phi$ as a parameter-free formula over $\bigcup_{i=1}^{k+l}\progVars_i$. 
	In this case, most of our steps correspond to the computation of the strongest postcondition \cite{0067387,SousaD16,ChenFD17} in a purely universal ($k$-safety) setting. \demo
\end{remark}

Our algorithm analyses the structure of each program and applies the insights from FEHL:
If $\univC$ and $\existsC$ are empty, we return $(\Phi, \top)$ (line \ref{line:done}), i.e., we do not place any restrictions on the parameters.
In case all programs are loops (line \ref{line:loops}), we invoke a subroutine \lstinline[style=code-style-large, language=code-lang]|genppLoops| (discussed in \Cref{sec:algorithmDetailsLoop}).
Otherwise, some program has a non-loop statement at the top level, allowing further symbolic analysis. 
We consider possible steps in $\univC$ (lines \ref{line:start-univ}-\ref{line:end-univ}) and in $\existsC$  (lines \ref{line:start-exists}-\ref{line:end-exists}).

We first consider the case where a universally quantified program has a non-loop statement at its top level (lines \ref{line:start-univ}-\ref{line:end-univ}).
In lines \ref{line:univ-skip1}, \ref{line:univ-skip2}, \ref{line:univ-assoc},  and  \ref{line:univ-skip-intro}, we bring the first program into the form $\prog_1 \mysemi \prog_2$ where $\prog_1 \neq \_ \mysemi \_$ by potentially inserting $\myskip$ statements in line \ref{line:univ-skip-intro}.
For a program $x \myassign e\mysemi \prog$ (line \ref{line:univ-assign}), we use \nameref{rule:forall-step} to handle the assignment. 
Here, we can compute the strongest postcondition of the assignment as $\exists x'. \Phi[x'/x] \land x = e[x'/x]$ (using Floyd's forward running rule \cite{floyd1993assigning}).
For conditionals (line \ref{line:univ-if}), we analyze both branches under the strengthened precondition.
As our analysis operates on parametric assertions, some of the parameters found in the precondition $\Phi$ can be restricted in \emph{both branches}.  
After we have computed a parametric postcondition for each branch, we therefore combine them into a parametric postcondition for the entire program by constructing the disjunction of the function-formulas $\postF_1$ and $\postF_2$ (describing the set of states reachable in either of the branches), and conjoining the restriction-formulas $\condF_1$ and $\condF_2$. 
For assume statements (line \ref{line:univ-assume}), we strengthen the precondition.
For nondeterministic assignments $x \myassign \star$ (line \ref{line:univ-nd}), we invalidate all knowledge about $x$.
If a program matches none of the previous cases (line \ref{line:univ-reorder}), it must be of the form $\mywhile(\_, \_)\mysemi \_$, and we move it to the end of $\univC$, continuing the analysis of the renaming programs in the next recursive iteration. 
If no universally quantified program can be analyzed further, we continue the investigation with existentially quantified ones (lines \ref{line:start-exists}-\ref{line:end-exists}).
Many cases are analogous to the treatment in universally quantified programs (lines \ref{line:same-cases-start}-\ref{line:same-cases-done}), but some cases are handled fundamentally differently:
If we encounter an assume statement $\myassume(b)$ (line \ref{line:exists-assume}), we need to certify that $b$ holds in all states in $\Phi$ (cf.~\nameref{rule:exists-assume}).
As we already hinted in \Cref{ex:para-assertion}, we accomplish this by restricting the viable set of parameters in $\Phi$, i.e., we restrict the domain of the function formula $\Phi$.
Concretely, we consider the formula $\condF_\mathit{assume} := \bigforall_{x \in \progVars_1 \cup \cdots \cup \progVars_{k+l} } x \ldot  (\Phi \Rightarrow b)$ (which is a formula over $\parameter$) that characterizes exactly those parameters that ensure that all states in $\Phi$ satisfy $b$.
After analyzing the remaining programs, we then conjoin $\condF_\mathit{assume}$ with the remaining restrictions.

\begin{remark}
	As in \Cref{fnlabel}, we can consider the case where $\Phi$ contains no parameter.
	In this case, $\condF_\mathit{assume}$ is a variable-free formula that is equivalent to $\top$ iff all states in $\Phi$ satisfy $b$.
	If $\Phi$ does \emph{not} imply $b$ (so $\condF_\mathit{assume}\equiv \bot$), the resulting parametric postcondition thus cannot prove any FEHT via \Cref{prop:postImplication}.\demo
\end{remark}

For nondeterministic assignments $x \myassign \star$ (line \ref{line:exists-nd}), we create a fresh parameter $\mu$ and continue the analysis under the precondition that $x = \mu$, effectively postponing the choice of a concrete value for $x$ (cf.~\Cref{ex:para-assertion}).

\begin{example}\label{ex:full}
	Our algorithm will automatically compute the parametric postcondition from \Cref{ex:para-assertion}.
	In particular, for the $\myassume(y \geq 2)$ statement, we match line \ref{line:exists-assume} with $\Phi = x \geq 9 \land y = \mu$ for $\mu \in \parameter$ and compute $\condF_\mathit{assume} := \forall x, y\ldot \Phi \Rightarrow y \geq 2$, which is logically equivalent to $\mu \geq 2$. \demo
\end{example}

\begin{algorithm}[!t]
	\caption{Parametric postcondition generation for loops
	}\label{algo:loop}
\vspace{-2mm}
\begin{code}
def genppLoops($\Phi$,$\bigoast_{i=1}^k \big(\mywhile(b_i, \prog_i)\mysemi \progg_i\big)$,$\bigoast_{i=k+1}^{k+l} \big(\mywhile(b_i, \prog_i)\mysemi \progg_i\big)$):
@@@$\mathbb{I}, c_1, \ldots, c_{k+l}$ := guessInvariantAndCounts()
@@@$B$ := max($c_1$,$\ldots$,$c_{k+l}$)
@@@$\condF_\mathit{init}$ := $\bigforall_{x \in \progVars_1 \cup \cdots \cup \progVars_{k+l} } x \ldot (\Phi \Rightarrow \mathbb{I})$(*\label{line:loop-cond1}*)
@@@$\condF_\mathit{sim}$ := $\bigforall_{x \in \progVars_1 \cup \cdots \cup\progVars_{k+l} } x \ldot (\mathbb{I} \Rightarrow \bigwedge_{i = 2}^{k+l} b_1\leftrightarrow b_i)$(*\label{line:loop-cond2}*)
@@@$\postF_1$ := $\mathbb{I}$
@@@for $j$ from $1$ to $B$: 
@@@@@@($\postF_{j+1}, \condF_{j+1}$) := genpp($\postF_j \land   \bigwedge_{i = 1 \mid c_i \geq j}^{k+l} b_i$,$\bigoast_{\substack{i=1 \mid c_i \geq j}}^{k}    \prog_i$,$\bigoast_{\substack{i={k+1} \mid c_i \geq j}}^{k+l}  \prog_i$) (*\label{line:loop-bodies}*)
@@@@@@$\condF_{j+1}^\mathit{cont}$ := $\bigforall_{x \in \progVars_1 \cup \cdots \cup\progVars_{k+l} } x \ldot (\postF_{j+1} \Rightarrow \bigwedge_{i=1 \mid c_i > j}^{k+l} b_i)$ (*\label{line:loop-no-term}*)
@@@$\condF_\mathit{ind}$ := $\bigforall_{x \in \progVars_1 \cup \cdots \cup\progVars_{k+l} } x \ldot (\postF_{B+1} \Rightarrow \mathbb{I})$(*\label{line:loop-ind}*)
@@@$(\postF_\mathit{rem}, \condF_\mathit{rem})$ := genpp($\mathbb{I} \land \bigwedge_{i=1}^{k+l} \neg b_i$,$\bigoast_{\substack{i=1 }}^{k}  \progg_i$, $\bigoast_{\substack{i=k+1}}^{k+l}  \progg_i$)
@@@return ($\postF_\mathit{rem}$,$\condF_\mathit{init} \land \condF_\mathit{sim} \land \bigwedge_{j=2}^{B+1} \condF_j \land \bigwedge_{j=2}^{B+1} \condF_{j}^\mathit{cont} \land \condF_\mathit{ind} \land \condF_\mathit{rem}$)
\end{code}
\vspace{-2mm}
\end{algorithm}

\subsection{Generating Parametric Postconditions for Loops}\label{sec:algorithmDetailsLoop}

We sketch the postcondition generation for loops in \Cref{algo:loop}.
As input, \lstinline[style=code-style-large, language=code-lang]|genppLoops| expects a precondition $\Phi$ over $\bigcup_{i=1}^{k+l}\progVars_i \cup \parameter$  and universally and existentially quantified loop programs. 
In the first step, we guess a loop invariant $\mathbb{I}$ and counter values $c_1, \ldots, c_{k+l} \in [1, B]$ (cf.~\nameref{rule:loop-count}).
In lines \ref{line:loop-cond1} and \ref{line:loop-cond2}, we ensure that $\mathbb{I}$ is initial and guarantees simultaneous termination by computing restrictions $\condF_\mathit{init}$ and $\condF_\mathit{sim}$ on the parameters present in $\Phi$ (similar to \myassume{} statements in line \ref{line:exists-assume} of \Cref{algo:verify}).
Again, in the special case where $\Phi$ contains no parameter (as is, e.g., the case when applying our algorithm to $k$-safety properties), $\condF_\mathit{init}$ (resp.~$\condF_\mathit{sim}$) is equivalent to $\top$ iff the invariant is initial (resp.~guarantees simultaneous termination).
Afterward, we check the validity of the guessed counter values $c_1, \ldots, c_{k+l}$. 
For each $j$ from $1$ to $B$, we compute a parametric postcondition $(\postF_{j+1}, \condF_{j+1})$ for the bodies of all loops that should be executed at least $j$ times (i.e., $c_i \geq j$) starting from precondition $\postF_j$ via a (mutually recursive) call to \lstinline[style=code-style-large, language=code-lang]|genpp| (line \ref{line:loop-bodies}). 
To ensure valid derivation using \nameref{rule:loop-count} we need to ensure that -- in $\postF_{j+1}$ -- the guard of all loops that we want to execute \emph{more} than $j$ times still evaluates to true. 
We ensure this by computing the restriction-formula $\condF_{j+1}^\mathit{cont}$, which restricts the parameters (both those already present in the precondition $\Phi$ and those added during the analysis of the loop bodies) such that all states in $\postF_{j+1}$ fulfill the guards of all loops with $c_i > j$ (line \ref{line:loop-no-term}).
After we have symbolically executed all loops the desired number of times, we construct a parameter restriction $\condF_\mathit{ind}$ that ensures that we end within the invariant, i.e., $\postF_{B+1} \Rightarrow \mathbb{I}$ (line \ref{line:loop-ind}).
In the last step, we compute a parametric postcondition $(\postF_{\mathit{rem}}, \condF_{\mathit{rem}})$ for the program suffix executed after the loops. 
We return the parametric postcondition that consists of the function-formula $\postF_{\mathit{rem}}$ and the conjunction of all restriction-formulas.

\subsection{The Main Verification}

From the soundness of FEHL (\Cref{theo:sound}) we directly get:

\begin{proposition}\label{lem:correct-pp}
	\lstinline[style=code-style-large, language=code-lang]|genpp($\Phi$,$\univC$,$\existsC$)| computes some parametric postcondition for $(\Phi, \univC, \existsC)$.
\end{proposition}

Given an FEHT $\rel{\Phi}{\univC}{\existsC}{\Psi}$, we can thus invoke \lstinline[style=code-style-large, language=code-lang]|genpp($\Phi$,$\univC$,$\existsC$)| to compute a parametric postcondition, which (if strong enough) allows us to prove that $\rel{\Phi}{\univC}{\existsC}{\Psi}$ is valid via \Cref{prop:postImplication}.
If the postcondition is too weak, we can re-run \lstinline[style=code-style-large, language=code-lang]|genpp| using updated invariant guesses (cf.~\Cref{sec:implementation}).
For loop-free programs, it is easy to see that \lstinline[style=code-style-large, language=code-lang]|genpp| computes the ``strongest possible`` parametric postcondition (it effectively executes the programs symbolically without incurring the imprecision inserted by loop invariants).
In this case, the query from \Cref{prop:postImplication} holds if and only if the FEHT is valid; our algorithm thus constitutes a \emph{complete} verification method.

\paragraph{Invalid FEHTs.}

We stress that the goal of our algorithm is the verification of FEHTs and not proving that an FEHT is \emph{invalid}.
For $k$-safety properties, a refutation (counterexample) consists of a $k$-tuple of concrete executions that violate the property \cite{SousaD16,ChenFD17}.
In contrast, refuting an FEHT corresponds to \emph{proving} a $\exists^*\forall^*$ property, an orthogonal problem that requires independent proof ideas.

\section{Implementation and Experiments}\label{sec:implementation}

We have implemented our verification algorithm in a tool called \tool{} \cite{tool} (short for \textbf{For}all \textbf{Ex}ists Verification), supporting programs in a minimalistic \texttt{C}-like language that features basic control structures (cf.~\Cref{sec:prelim}), arrays, and bitvectors.
\tool{} uses \texttt{Z3} \cite{MouraB08} to discharge SMT queries and supports the theory of linear integer arithmetic, the theory of arrays, and the theory of finite bitvectors.
Compared to the presentation in \Cref{sec:algorithm}, we check satisfiability of restriction-formulas \emph{eagerly}:
For example, in \Cref{algo:loop}, we compute multiple restriction-formulas and return their conjunction. 
In \tool{}, we immediately check these intermediate restrictions for satisfiability; if any restriction is unsatisfiable on its own, any conjunction involving it will be as well, so we can abort the analysis early and re-start parts of the analysis using, e.g., updated invariants and counter values.

\subsection{Loop Invariant Generation}\label{sec:sub:inv}

Our loop invariant generation and counter value inference follows a standard guess-and-check procedure \cite{FlanaganL01,SharmaGHALN13,SousaD16,ChenFD17,SharmaA14}, i.e., we generate promising candidates by combining expressions found in the programs and equalities between variables in the loop guards.
In most loops, there exist ``anchor'' variables that effectively couple executions of multiple loops together \cite{SousaD16,ChenFD17}; even in asynchronous cases like \Cref{ex:counting-example}.
Exploring more advanced invariant generation techniques is interesting future work.
However -- even in the simpler setting of $k$-safety properties -- many tools currently rely on a guess-and-check approach \cite{SousaD16,ChenFD17}.
We maintain a lattice of possible candidates ordered by implication, which allows us for efficient pruning.
For example, if the current candidate is not initial (i.e., $\condF_\mathit{init}$ computed in line \ref{line:loop-cond1} of \Cref{algo:loop} is unsatisfiable), we do not need to consider stronger candidates. 
Likewise, if the candidate does not ensure simultaneous termination ($\condF_\mathit{sim}$) we can prune all weaker invariants.

\subsection{Experiments}\label{sec:sub:experiments}

We evaluate \tool{} in various settings where FEHT-like specifications arise.
We compare with \HyPA{} (a predicate-abstraction-based solver) \cite{BeutnerF22}, \PCSAT{} (a constraint-based solver that relies on predicate templates) \cite{UnnoTK21}, and \HyPro{} (a model-checker for $\forall^*\exists^*$ properties in \emph{finite-state} systems) \cite{BeutnerF22b}.
Our results were obtained on a M1 Pro CPU with 32GB of memory.

\begin{figure}[!t]
	
	\begin{minipage}[t]{0.5\linewidth}
		\begin{subtable}{\linewidth}
			
			\centering
			\scalebox{1.0}{\parbox{\linewidth}{
					\small
				\centering
				\def\arraystretch{1.1}
				\begin{tabular}{l@{\hspace{5mm}}c@{\hspace{5mm}}c}
					\toprule
					\textbf{Instance} & $\boldsymbol{t}_{\HyPA{}}$ & $\boldsymbol{t}_{\tool{}}$ \\
					\midrule
					\textsc{DoubleSquareNI}$^{\dagger}$  &  67.12 & \textbf{0.71} \\
					\textsc{Exp1x3}  &3.79 & \textbf{0.30} \\
					\textsc{Fig3} &  8.78 & \textbf{0.39} \\
					\textsc{DoubleSquareNIff}  &  4.91  & \textbf{0.37} \\
					\textsc{Fig2}$^{\dagger}$  &  17.7 & \textbf{0.73} \\
					\textsc{ColIitemSymm}  & 15.51 & \textbf{0.20} \\
					\textsc{CounterDet}  & 5.28 & \textbf{0.55}\\
					\textsc{MultEquiv} & 13.13 & \textbf{0.60} \\
					\textsc{HalfSquareNI} & \textbf{68.04}  & - \\
					\textsc{SquaresSum} & \textbf{17.03} & - \\
					\textsc{ArrayInsert} & \textbf{16.17} & - \\
					\bottomrule
				\end{tabular}
			}}
			\vspace{-0mm}
			\subcaption{}\label{tab:k-safety}
		\end{subtable}
	\end{minipage}%
	\begin{minipage}[t]{0.5\linewidth}
		\begin{subtable}{\linewidth}
				\centering
				\scalebox{1.0}{\parbox{\linewidth}{
						\centering
						\small
						\def\arraystretch{1.1}
						\begin{tabular}{l@{\hspace{5mm}}c@{\hspace{5mm}}c}
							\toprule
							\textbf{Instance} & $\boldsymbol{t}_{\HyPA{}}$ & $\boldsymbol{t}_{\tool{}}$  \\
							\midrule
							\textsc{NonDetAdd} &  3.63 & \textbf{0.76} \\
							\textsc{CounterSum} &  5.05 & \textbf{1.95} \\
							\textsc{AsynchGNI} &   5.20 & \textbf{0.69} \\
							\textsc{CompilerOpt1} &  1.79 & \textbf{0.59} \\
							\textsc{CompilerOpt2} &  2.71 & \textbf{1.02} \\
							\textsc{Refine} &  10.1 & \textbf{0.57} \\
							\textsc{Refine2} &  9.87 & \textbf{0.64} \\
							\textsc{Smaller} &  2.21 & \textbf{0.69} \\
							\textsc{CounterDiff}  &  8.05 & \textbf{0.63} \\
							\textsc{Fig.~3}  & 8.92 & \textbf{0.57} \\
							\bottomrule
						\end{tabular}
				}}
			
			\vspace{-0mm}
			\subcaption{} \label{tab:hypa}
		\end{subtable}
	\end{minipage}\\[4mm]
	\begin{minipage}[t]{0.5\linewidth}
		\begin{subtable}{\linewidth}
				\centering
				\scalebox{1.0}{\parbox{\linewidth}{
						\small
						\centering
						\def\arraystretch{1.1}
						\begin{tabular}{l@{\hspace{5mm}}c@{\hspace{5mm}}c}
							\toprule
							\textbf{Instance} & $\boldsymbol{t}_{\PCSAT{}}$ & $\boldsymbol{t}_{\tool{}}$ \\
							\midrule
							\textsc{TI\_GNI\_hFF}  &  26.2 & \textbf{0.58} \\
							\textsc{TI\_GNI\_hTT}  &  32.5 & \textbf{0.10} \\
							\textsc{TI\_GNI\_hFT}$^{\dagger, \ddagger}$ &  36.2 & \textbf{0.70} \\
							\textsc{TS\_GNI\_hFF}  &  36.6 & \textbf{0.58} \\
							\textsc{TS\_GNI\_hTT}$^{\ddagger}$  &  96.2 & \textbf{0.16} \\
							\textsc{TS\_GNI\_hFT}$^{\dagger, \ddagger}$  &  123.3 & \textbf{2.88} \\
							\textsc{TI\_GNI\_hTF}\!\!  & \textbf{26.1} & -\\
							\textsc{TS\_GNI\_hTF} & \textbf{44.1} & - \\
							\bottomrule
						\end{tabular}
				}}
			\subcaption{}\label{tab:constraintComp}
		\end{subtable}
	\end{minipage}%
	\begin{minipage}[t]{0.5\linewidth}
		
		\begin{subfigure}{0.9\linewidth}
			\begin{center}
				\includegraphics[width=\linewidth]{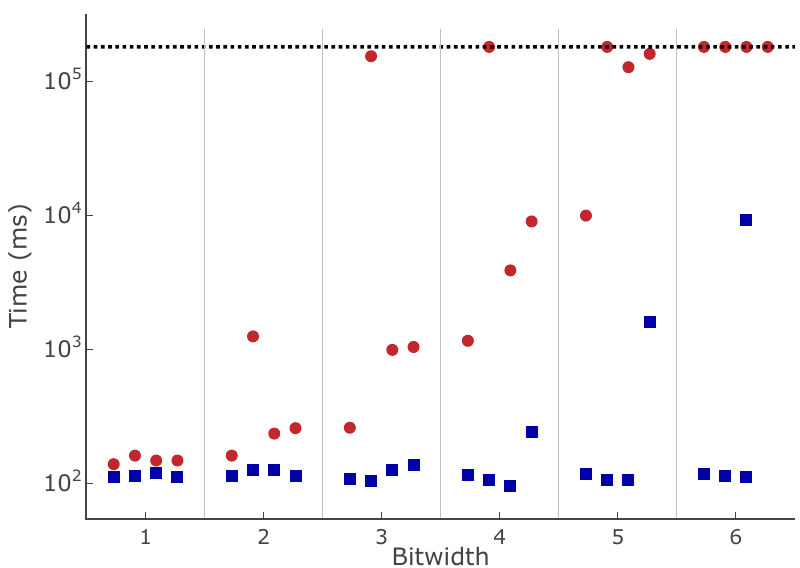}
			\end{center}

			\vspace{-2mm}
			\subcaption{}\label{fig:explicit}
		\end{subfigure}
		
	\end{minipage}
		
	\caption{In Tables \ref{tab:k-safety} and \ref{tab:hypa}, we compare \tool{} with \HyPA{} \cite{BeutnerF22} on $k$-safety and $\forall^*\exists^*$ properties, respectively.
		For instances marked with $\dagger$, \tool{} required additional user-provided invariant hints. 
		 In Table \ref{tab:constraintComp}, we compare \tool{} with \PCSAT{} \cite{UnnoTK21}.
		 For instances marked with $\ddagger$, \PCSAT{} required additional invariant hints. 
		In \Cref{fig:explicit}, we compare the running time of \tool{} (\scalebox{0.8}{\color{plotBlue}$\blacksquare$}) and \HyPro{} \cite{BeutnerF22b} (\scalebox{1.4}{\color{plotRed}$\bullet$}). We check each of the 4 GNI instances from \cite{BeutnerF22b} with varying bitwidth. The timeout is set to 3 min (marked by the horizontal dotted line). }
\end{figure}

\paragraph{Limitations of \tool{}'s Loop Alignment.}

Before we evaluate \tool{} on $\forall^*\exists^*$ properties, we investigate the counting-based loop alignment principle underlying \tool{}. 
We collect the $k$-safety benchmarks from \HyPA{} \cite{BeutnerF22} (which themself were collected from multiple sources \cite{FarzanV20,FarzanV19,ShemerGSV19,UnnoTK21}) and depict the verification results in Table \ref{tab:k-safety}. 
We observe that \tool{} can verify many of these instances.
As it explores a restricted class of loop alignments (guided by \nameref{rule:loop-count}), it is more efficient on the instances it can solve. 
However, for some of the instances, \tool{}'s counting-based alignment is insufficient.
Instead, these instances require a loop alignment that is context-dependent, i.e., the alignment is chosen based on the current state of the programs \cite{BeutnerF22,ShemerGSV19,FarzanV20,UnnoTK21}. 

\paragraph{\tool{} and \HyPA{}.}

\HyPA{} \cite{BeutnerF22} explores a liberal program alignment by exploring  a user-provided predicate abstraction.
The verification instances considered in \cite{BeutnerF22} include a range of $\forall^*\exists^*$ properties on very small programs, including, e.g., GNI and refinement properties.
In Table \ref{tab:hypa}, we compare the running time of \tool{} with that of \HyPA{} (using the user-defined predicates for its abstraction).\footnote{The properties checked by \HyPA{} \cite{BeutnerF22} are temporal, i.e., properties about the infinite execution of programs of the form $\mywhile(\top, \prog)$.
	To make such programs analyzable in \tool{} (which reasons about finite executions), we replaced the \emph{infinite} loop with a loop that executes $\prog$ some fixed (but arbitrary) number of times.}
We observe that \tool{} can verify the instances significantly quicker. 
Moreover, we stress that \tool{} solves a much more challenging problem as it analyzes the program \emph{fully automatically} without any user intervention.

\paragraph{\tool{} and \PCSAT{}.}

Unno et al.~\cite{UnnoTK21} present an extension of constraint Horn clauses, called pfwCSP, that is able to express a range of relational properties (including $\forall^*\exists^*$ properties). 
Their custom pfwCSP solver (called \PCSAT{}) instantiates predicates with user-provided templates. 
We compare \PCSAT{} and \tool{} in Table \ref{tab:constraintComp}.
\tool{} can verify 6 out of the 8 $\forall^*\exists^*$ instances.
\tool{} currently does not support termination proofs for loops in existentially quantified programs (which are needed for \textsc{TI\_GNI\_hTF} and \textsc{TS\_GNI\_hTF}), whereas \PCSAT{} features loop variant templates and can thus reason about the termination of existentially quantified loops in isolation. 
In the instances that \tool{} can solve, it is much faster. 
We conjecture that this is due to the fact that the constraints generated by \tool{} can be solved directly by SMT solvers, whereas \PCSAT{}'s pfwCSP constraints first require a custom template instantiation.

\paragraph{\tool{} and \HyPro{}.}

Programs whose variables have a finite domain (e.g., boolean) can be checked using explicit-state techniques developed for logics such as HyperLTL \cite{ClarksonFKMRS14}.
We verify GNI on variants of the four boolean programs from \cite{BeutnerF22b} with a varying number of bits.
We compare \tool{} with the HyperLTL verifier \HyPro{} \cite{BeutnerF22b}, which converts a program into an explicit-state transition system. 
We depict the results in \Cref{fig:explicit}.
We observe that, with increasing bitwidth, the running time of explicit-state model-checking increases exponentially (note that the scale is logarithmic).
In contrast, \tool{} can employ symbolic bitvector reasoning, resulting in orders of magnitude faster verification.

\section{Related Work}\label{sec:relatedWork}

Most methods for $k$-safety verification are centered around the self-composition of a program \cite{BartheDR11} and often improve upon a na\"ive self-composition by, e.g., exploiting  the commutativity of statements \cite{ShemerGSV19,FarzanV19,FarzanV20,EilersMH20}.
Relational program logics for $k$-safety offer a rich set of rules to \emph{over}-approximate the program behavior \cite{Benton04,Yang07,SousaD16,NagasamudramN21,DOsualdoFD22,AntonopoulosKLNNN23,Beringer11}.
Recently, much effort has been made to employ under-approximate methods that find bugs instead of proving their absence; so far, mostly for unary (non-hyper) properties \cite{OHearn20,VriesK11,RaadBDDOV20,MollerOH21,MaksimovicCLSG23,BruniGGR23,ZilbersteinDS23,Cousot24}.

Dardinier et al.~\cite{DardinierM23} propose \emph{Hyper Hoare Logic} -- a logic that can express \emph{arbitrary} hyperproperties, but requires manual deductive reasoning.
Dickerson et al.~\cite{DickersonYZD22} introduce RHLE, a program logic for the verification of $\forall^*\exists^*$ properties, focusing on the composition (and under-approximation) of function calls.
They present a weakest-precondition-based verification algorithm that aligns loops in lock-step via user-provided loop invariants. 
Unno et al.~\cite{UnnoTK21} present an extension of constraint Horn-clauses (called pfwCSP).
They show that pfwCSP can encode many relational verification conditions, including many hyperliveness properties like GNI (see \Cref{sec:implementation}). 
Compared to the pfwCSP encoding, we explore a less liberal program alignment (guided by \nameref{rule:loop-count}).
However, we gain the important advantage of generating standard (first-order) SMT constraints that can be handled using existing SMT solvers (which shows significant performance improvement, cf.~\Cref{sec:implementation}).

Most work on the verification of hyperliveness has focused on more general \emph{temporal} properties, i.e., properties that reason about infinite executions, based on logics such as HyperLTL \cite{ClarksonFKMRS14,FinkbeinerRS15,BeutnerF23}.
Coenen et al.~\cite{CoenenFST19} study a method for verifying hyperliveness  in \emph{finite}-state transition systems using strategies to resolve existential quantification.
This approach is also applicable to infinite-state systems by means of an abstraction \cite{BeutnerF22,ItzhakySV23} (see \HyPA{} in \Cref{sec:implementation}).
Bounded model-checking (BMC) for hyperproperties \cite{HsuSB21} unrolls the system to a fixed bound and can, e.g., find violations to GNI.
Existing BMC tools target finite-state (boolean) systems and construct QBF formulas; lifting this to support infinite-state systems by constructing SMT constraints is an interesting future work and could, e.g., complement \tool{} in the refutation of FEHTs.

\section{Conclusion}

We have studied the automated program verification of relational $\forall^*\exists^*$ properties.
We developed a constraint-based verification algorithm that is rooted in a sound-and-complete program logic and uses a (parametric) postcondition computation.
Our experiments show that -- while our logic-guided tool explores a restricted class of possible loop alignments -- it succeeds in many of the instances we tested.
Moreover, the use of off-the-shelf SMT solvers results in faster verification, paving the way toward a future of fully automated tools that can check important hyperliveness properties such as GNI and opacity.

\subsubsection*{Acknowledgments.}

This work was supported by the European Research Council (ERC) Grant HYPER (101055412), and by the German Research Foundation (DFG) as part of TRR 248 (389792660).

\subsubsection*{Data Availability Statement.}

\tool{} is available at \cite{tool}.

\bibliographystyle{splncs04}
\bibliography{references}

\iffullversion

\newpage

\appendix

\section{Program Semantics}\label{app:semantics}

For a state $\sigma$, variable $x \in \progVars$, and $z \in \intSet$ we write $\sigma[x \mapsto z]$ for the state in which we update the value of $x$ to $z$.
Given an arithmetic (resp.~boolean) expression $e$ (resp.~$b$), we write $\bsem{e}(\sigma) \in \intSet$ (resp.~$\bsem{b}(\sigma) \in \bool$) for the value of this expression in $\sigma$.
For a program $\prog$, we define the semantics $\sem{\prog} \subseteq (\progVars \to \intSet) \times (\progVars \to \intSet)$ inductively using the rules in \Cref{fig:sem}.

\begin{figure}[!t]
	\begin{minipage}{0.4\textwidth}
		\begin{prooftree}
			\AxiomC{}
			\UnaryInfC{$\sem{x \myassign e}(\sigma, \sigma[x \mapsto \bsem{e}(\sigma)])$}
		\end{prooftree}
	\end{minipage}%
	\begin{minipage}{0.3\textwidth}
		\begin{prooftree}
			\AxiomC{$z \in \intSet$}
			\UnaryInfC{$\sem{x \myassign \star}(\sigma, \sigma[x \mapsto z])$}
		\end{prooftree}
	\end{minipage}%
	\begin{minipage}{0.3\textwidth}
		\begin{prooftree}
			\AxiomC{$\bsem{b}(\sigma) = \top$}
			\UnaryInfC{$\sem{\myassume(b)}(\sigma, \sigma)$}
		\end{prooftree}
	\end{minipage}
	
	\vspace{5mm}
	
	\begin{minipage}{0.5\textwidth}
		\begin{prooftree}
			\AxiomC{}
			\UnaryInfC{$\sem{\myskip}(\sigma, \sigma)$}
		\end{prooftree}
	\end{minipage}%
	\begin{minipage}{0.5\textwidth}
		\begin{prooftree}
			\AxiomC{$\sem{\prog}(\sigma, \sigma')$}
			\AxiomC{$\sem{\progg}(\sigma', \sigma'')$}
			\BinaryInfC{$\sem{\prog \mysemi \progg}(\sigma, \sigma'')$}
		\end{prooftree}
	\end{minipage}%
	\vspace{5mm}

	\begin{minipage}{0.5\textwidth}
		\begin{prooftree}
			\AxiomC{$\bsem{b}(\sigma) = \top$}
			\AxiomC{$\sem{\prog}(\sigma, \sigma')$}
			\BinaryInfC{$\sem{\myif(b, \prog, \progg)}(\sigma, \sigma')$}
		\end{prooftree}
	\end{minipage}%
	\begin{minipage}{0.5\textwidth}
		\begin{prooftree}
			\AxiomC{$\bsem{b}(\sigma) = \bot$}
			\AxiomC{$\sem{\progg}(\sigma, \sigma')$}
			\BinaryInfC{$\sem{\myif(b, \prog, \progg)}(\sigma, \sigma')$}
		\end{prooftree}
	\end{minipage}
	
	\vspace{5mm}
	
	\begin{minipage}{0.3\textwidth}
		\begin{prooftree}
			\AxiomC{$\bsem{b}(\sigma) = \bot$}
			\UnaryInfC{$\sem{\mywhile(b, \prog)}(\sigma, \sigma)$}
		\end{prooftree}
	\end{minipage}%
	\begin{minipage}{0.7\textwidth}
		\begin{prooftree}
			\AxiomC{$\bsem{b}(\sigma) = \top$}
			\AxiomC{$\sem{\prog}(\sigma, \sigma')$}
			\AxiomC{$\sem{\mywhile(b, \prog)}(\sigma', \sigma'')$}
			\TrinaryInfC{$\sem{\mywhile(b, \prog)}(\sigma, \sigma'')$}
		\end{prooftree}
	\end{minipage}

	\vspace{3mm}
	
	\caption{Operational program semantics}\label{fig:sem}
\end{figure}

\section{Details on FEHL}\label{app:fehl}

\subsection{Full Collection of Proof Rules}

\begin{figure}[!t]
	\def\ScoreOverhang{2pt}
	\footnotesize
	\vspace{-3mm}
	\begin{minipage}{0.33\linewidth}
		\centering
		\scalebox{0.85}{\parbox{\linewidth}{
				\begin{topprooftree}{\labelText{\texttt{($\forall$-Reorder)}}{rule:forall-comm}
						\labelTextGrey{\texttt{($\exists$-Reorder)}}{rule:exists-comm}}
					\AxiomC{$\vdash\rel{\Phi}{\univC_2 \oast \univC_1}{\existsC}{\Psi}$}
					\UnaryInfC{$\vdash\rel{\Phi}{\univC_1 \oast \univC_2}{\existsC}{\Psi}$}
				\end{topprooftree}
		}}
	\end{minipage}%
	\begin{minipage}{0.33\linewidth}
		\centering
		\scalebox{0.85}{\parbox{\linewidth}{
				\begin{topprooftree}{\labelText{\texttt{($\forall$-Skip-I)}}{rule:forall-intro} \labelTextGrey{\texttt{($\exists$-Skip-I)}}{rule:exists-intro}}
					\AxiomC{$\vdash\rel{\Phi}{\prog\mysemi\myskip \oast \univC}{\existsC}{\Psi}$}
					\UnaryInfC{$\vdash\rel{\Phi}{\prog \oast \univC}{\existsC}{\Psi}$}
				\end{topprooftree}
		}}
	\end{minipage}%
	\begin{minipage}{0.33\linewidth}
		\centering
		\scalebox{0.85}{\parbox{\linewidth}{
				\begin{topprooftree}{\labelText{\texttt{($\forall$-Skip-E)}}{rule:forall-elim} \labelTextGrey{\texttt{($\exists$-Skip-E)}}{rule:exists-elim}}
					\AxiomC{$\vdash \rel{\Phi}{\univC}{\existsC}{\Psi}$}
					\UnaryInfC{$\vdash \rel{\Phi}{\myskip \oast \univC}{\existsC}{\Psi}$}
				\end{topprooftree}
		}}
	\end{minipage}%

	\vspace{-2mm}
	
	\begin{minipage}{0.2\linewidth}
		\centering
		\scalebox{0.85}{\parbox{\linewidth}{
				\begin{topprooftree}{\labelText{\texttt{(Done)}}{rule:done}}
					\AxiomC{$\vdash \rel{\Phi}{\epsilon}{\epsilon}{\Phi}$}
				\end{topprooftree}
		}}
	\end{minipage}
	\begin{minipage}{0.4\linewidth}
		\vspace{3mm}
		\centering
		\scalebox{0.85}{\parbox{\linewidth}{
				\def\defaultHypSeparation{\hskip .15in}
				\begin{topprooftree}{\labelText{\texttt{(Cons)}}{rule:cons}}
					\AxiomC{\stackanchor{$\Phi \Rightarrow \Phi'$}{$\Psi' \Rightarrow \Psi$}}
					\AxiomC{$\vdash\rel{\Phi'}{\univC}{\existsC}{\Psi'}$}
					\BinaryInfC{$\vdash\rel{\Phi}{\univC}{\existsC}{\Psi}$}
				\end{topprooftree}
		}}
	\end{minipage}%
	\begin{minipage}{0.4\linewidth}
		\centering
		\scalebox{0.85}{\parbox{\linewidth}{
				\def\defaultHypSeparation{\hskip .2in}
				\begin{topprooftree}{\labelText{\texttt{($\forall$-If)}}{rule:forall-if-app}
						\labelTextGrey{\texttt{($\exists$-If)}}{rule:exists-if-app}}
					\AxiomC{\stackanchor{$\vdash\rel{\Phi \land b}{\prog_1\mysemi \prog_3 \oast \univC}{\existsC}{\Psi}$}{$\vdash\rel{\Phi \land \neg b}{\prog_2\mysemi \prog_3 \oast \univC}{\existsC}{\Psi}$}}
					\UnaryInfC{$\vdash\rel{\Phi}{\myif(b, \prog_1, \prog_2)\mysemi \prog_3 \oast \univC}{\existsC}{\Psi}$}
				\end{topprooftree}
		}}
	\end{minipage}%

	\vspace{-2mm}

	\begin{minipage}{0.33\linewidth}
		\centering
		\scalebox{0.85}{\parbox{\linewidth}{
		\begin{topprooftree}{\labelText{\texttt{($\forall$-$\mysemi$-Assoc)}}{rule:forall-seq-assoc}\labelTextGrey{\texttt{($\exists$-$\mysemi$-Assoc)}}{rule:exists-seq-assoc}}
			\AxiomC{$\vdash \rel{\Phi}{\prog_1 \mysemi (\prog_2\mysemi \prog_3) \oast \univC}{\existsC}{\Phi}$}
			\UnaryInfC{$\vdash \rel{\Phi}{(\prog_1 \mysemi \prog_2)\mysemi \prog_3 \oast \univC}{\existsC}{\Phi}$}
		\end{topprooftree}
	}}
	\end{minipage}%
	\begin{minipage}{0.33\linewidth}
		\vspace{1.5mm}
		\centering
		\scalebox{0.85}{\parbox{\linewidth}{
				\def\defaultHypSeparation{\hskip .2in}
				\begin{topprooftree}{\labelText{\texttt{($\forall$-Step)}}{rule:forall-step-app}}
					\AxiomC{\stackanchor{$\vdash\tripH{\Phi}{\prog_1}{\Phi'}$}{$\vdash\rel{\Phi'}{\prog_2 \oast \univC}{\existsC}{\Psi}$}}
					\UnaryInfC{$\vdash\rel{\Phi}{\prog_1 \mysemi \prog_2 \oast \univC}{\existsC}{\Psi}$}
				\end{topprooftree}
		}}
	\end{minipage}
	\begin{minipage}{0.33\linewidth}
		\centering
		\scalebox{0.85}{\parbox{\linewidth}{
				\def\defaultHypSeparation{\hskip .2in}
				\begin{topprooftree}{\labelText{\texttt{($\exists$-Step)}}{rule:exists-step-app}}
					\AxiomC{\stackanchor{$\vdash\tripN{\Phi}{\prog_1}{\Phi'}$}{$\vdash\rel{\Phi'}{\univC}{\prog_2 \oast \existsC}{\Psi}$}}
					\UnaryInfC{$\vdash\rel{\Phi}{\univC}{\prog_1\mysemi\prog_2 \oast \existsC}{\Psi}$}
				\end{topprooftree}
		}}
	\end{minipage}

	\vspace{-4mm}

	\begin{minipage}{0.5\linewidth}
		\vspace{3mm}
		\centering
		\scalebox{0.85}{\parbox{\linewidth}{
				\def\defaultHypSeparation{\hskip .2in}
				\begin{topprooftree}{\labelText{\texttt{($\forall$-Assume)}}{rule:forall-assume-app}}
					\AxiomC{$\vdash\rel{\Phi \land b}{\prog \oast \univC}{\existsC}{\Psi}$}
					\UnaryInfC{$\vdash\rel{\Phi}{\myassume(b) \mysemi \prog \oast \univC}{\existsC}{\Psi}$}
				\end{topprooftree}
		}}
	\end{minipage}%
	\begin{minipage}{0.5\linewidth}
		\centering
		\scalebox{0.85}{\parbox{\linewidth}{
				\def\defaultHypSeparation{\hskip .2in}
				\begin{topprooftree}{\labelText{\texttt{($\exists$-Assume)}}{rule:exists-assume-app}}
					\AxiomC{$\Phi \Rightarrow b$}
					\AxiomC{$\vdash\rel{\Phi}{\univC}{\prog \oast \existsC}{\Psi}$}
					\BinaryInfC{$\vdash\rel{\Phi}{\univC}{\myassume(b) \mysemi \prog \oast \existsC}{\Psi}$}
				\end{topprooftree}
		}}
	\end{minipage}%
	
	\vspace{-4mm}
	
	\begin{minipage}{0.4\linewidth}
		\vspace{3mm}
		\centering
		\scalebox{0.85}{\parbox{\linewidth}{
				\def\defaultHypSeparation{\hskip .05in}
				\begin{topprooftree}{\labelText{\texttt{($\forall$-Choice)}}{rule:forall-inf-nd-app}}
					\AxiomC{$\vdash\rel{\exists x\ldot \Phi}{\prog \oast \univC}{\existsC}{\Psi}$}
					\UnaryInfC{$\vdash\rel{\Phi}{x \myassign \star\mysemi \prog \oast \univC}{\existsC}{\Psi}$}
				\end{topprooftree}
		}}
	\end{minipage}
	\begin{minipage}{0.6\linewidth}
		\centering
		\scalebox{0.85}{\parbox{\linewidth}{
				\def\defaultHypSeparation{\hskip .2in}
				\begin{topprooftree}{\labelText{\texttt{($\exists$-Choice)}}{rule:exists-inf-nd-app}}
					\AxiomC{$x \not\in \mathit{Vars}(e)$}
					\AxiomC{$\vdash\rel{(\exists x\ldot \Phi) \land x = e}{\univC}{\prog \oast \existsC}{\Psi}$}
					\BinaryInfC{$\vdash\rel{\Phi}{\univC}{x \myassign \star\mysemi \prog  \oast \existsC}{\Psi}$}
				\end{topprooftree}
		}}
	\end{minipage}

	\vspace{-4mm}

	\begin{minipage}{0.4\linewidth}
		\vspace{3mm}
		\centering
		\scalebox{0.85}{\parbox{\linewidth}{
				\def\defaultHypSeparation{\hskip .05in}
				\begin{topprooftree}{\labelText{\texttt{($\forall$-Havoc)}}{rule:forall-havoc}}
					\AxiomC{$\vdash\rel{\bigexists_{x \in \mathit{ModVars}(\prog)} x\ldot \Phi}{\univC}{\existsC}{\Psi}$}
					\UnaryInfC{$\vdash\rel{\Phi}{\prog \oast \univC}{\existsC}{\Psi}$}
				\end{topprooftree}
		}}
	\end{minipage}
	\begin{minipage}{0.6\linewidth}
		\centering
		\scalebox{0.85}{\parbox{\linewidth}{
				\def\ScoreOverhang{0pt}
				\def\defaultHypSeparation{\hskip .1in}
				\begin{topprooftree}{\labelText{\texttt{($\exists$-Havoc)}}{rule:exists-havoc}}
					\AxiomC{$\vdash\tripN{\Phi}{\prog}{\top}$}
					\AxiomC{$\vdash \rel{\bigexists_{x \in \mathit{ModVars}(\prog)} x\ldot \Phi}{\univC}{\existsC}{\Psi}$}
					\BinaryInfC{$\vdash \rel{\Phi}{\univC}{\prog \oast \existsC}{\Psi}$}
				\end{topprooftree}
		}}
	\end{minipage}

	\vspace{-2mm}

	\caption{Full set of core proof rules of FEHL }\label{fig:core_rules_full}
\end{figure}

A full overview (subsuming those presented in \Cref{fig:core_rules}) of FEHL is given in \Cref{fig:core_rules_full}.
Rule \nameref{rule:forall-comm} allows the reordering of universally quantified programs (recall that all programs operate on disjoint variables).
There also exists an analogous rule \nameref{rule:exists-comm} that handles the reordering of existentially quantified programs; defined as expected.
We omit such obvious rule and only indict their existence by writing their name next to their universal counterpart. 
\nameref{rule:cons} corresponds to the standard structural rule of consequence which strengthens the precondition and weakens the postcondition.
\nameref{rule:forall-intro} and \nameref{rule:exists-intro} rewrite a program $\prog$ into $\prog \mysemi \myskip$.
Often, this rewrite is an important first step, as most other rules target programs of the form $\prog_1\mysemi\prog_2$.
Conversely, \nameref{rule:forall-elim} and \nameref{rule:exists-elim} eliminate programs consisting of a single $\myskip$ instruction.
Rules \nameref{rule:forall-seq-assoc} and \nameref{rule:exists-seq-assoc} exploit the associativity of sequential composition. 
Using these rules, we can always bring the form in the form $\prog_1 \mysemi \prog_2$ where $\prog_1 \neq \_ \mysemi \_$.
\nameref{rule:done} concludes any judgments in case we are considering the trivial $0$-fold self-composition with the empty list of programs on either side (denoted by $\epsilon$).

\paragraph{Havoc Rules.}
\nameref{rule:forall-havoc} and \nameref{rule:exists-havoc} allow us to skip the analysis of (potentially large) code fragments by invalidating all knowledge about variables modified within $\prog$.
Here we write $\mathit{ModVars}(\prog)$ for all variables that are changed in $\prog$, i.e., appear on the left-hand side of a deterministic or nondeterministic assignment. 
\nameref{rule:forall-havoc} removes an (arbitrary large) program at the expense of enlarging the precondition.
Its soundness can be argued easily:
If we take \emph{any} state $\sigma$ that satisfies $\Phi$ and consider any execution of $\prog$ from $\sigma$, $\prog$ will either not terminate (in which case the FEHT is vacuously valid as $\prog$'s execution are quantified universally), or it will terminate in some state $\sigma'$ that differs from $\sigma$ only in the values of $\mathit{ModVars}(\prog)$, and thus satisfies $\exists_{x \in \mathit{ModVars}(\prog)} x\ldot \Phi$. 
For existentially quantified copies, FEHTs postulate the \emph{existence} of a set of executions in $\prog$, which goes beyond partial correctness and requires us to reason about the termination of $\prog$.
Consequently, \nameref{rule:exists-havoc} adds an additional requirement compared to \nameref{rule:forall-havoc}:
In addition to invalidating all knowledge about modified variables, \nameref{rule:exists-havoc} requires us to prove the validity of the UHT $\tripN{\Phi}{\prog}{\top}$; effectively stating that $\prog$ \emph{can} terminate from all states in $\Phi$.
Without this additional restriction, we could could, e.g., derive the \emph{invalid} FEHT $\rel{\top}{\myskip}{\mywhile(\top, \myskip)}{\top}$.

\subsection{Soundness}

\soundness*
\begin{proof}
	We prove this statement by induction of the derivation of $\vdash \rel{\Phi}{\univC}{\existsC}{\Psi}$.
	We do a case analysis on the topmost rule. 
	Note that we omit obvious or analogous cases. 
	
	\begin{itemize}
		\item \nameref{rule:forall-comm}: By IH, we get that $\rel{\Phi}{\univC_2 \oast \univC_1}{\existsC}{\Psi}$ is valid.
		Now by the semantics of FEHTs, $\rel{\Phi}{\univC_2 \oast \univC_1}{\existsC}{\Psi}$ is valid iff $\rel{\Phi}{\univC_1 \oast \univC_2}{\existsC}{\Psi}$ is valid (as all programs operate on disjoint variables), as required.
		\item \nameref{rule:forall-intro}: By IH we have that $\rel{\Phi}{\prog\mysemi \myskip \oast \univC}{\existsC}{\Psi}$ is valid. 
		By the program semantics we have $\sem{\prog\mysemi \myskip} = \sem{\prog}$, and so the valdity $\rel{\Phi}{\prog \oast \univC}{\existsC}{\Psi}$ follows directly.
		\item \nameref{rule:forall-seq-assoc}: Follows directly as $\sem{(\prog_1 \mysemi \prog_2) \mysemi \prog_3} = \sem{\prog_1 \mysemi (\prog_2 \mysemi \prog_3)}$.
		\item \nameref{rule:cons}: By IH we get that $\rel{\Phi'}{\univC}{\existsC}{\Psi'}$ is valid. 
		We show that $\rel{\Phi}{\univC}{\existsC}{\Psi}$ is valid. 
		Take any states $\sigma_1, \ldots, \sigma_{k+l}$ such that $\bigoplus_{i=1}^{k+l} \sigma_i \models \Phi$, and any states $\sigma_1', \ldots, \sigma_k'$ such that $\sem{\prog_i}(\sigma_i, \sigma_i')$ for all $i \in [1, k]$ (as in the presumption of FEHT validity).
		As $\Phi \Rightarrow \Phi'$, we get that $\bigoplus_{i=1}^{k+l} \sigma_i \models \Phi'$. 
		As $\rel{\Phi'}{\univC}{\existsC}{\Psi'}$ is valid there thus exist final states $\sigma_{k+1}', \ldots, \sigma_{k+l}'$ such that $\sem{\prog_i}(\sigma_i, \sigma_i')$ for all $i \in [k+1, k+l]$ and $\bigoplus_{i=1}^{k+l} \sigma_i' \models \Psi'$.
		As $\Psi' \Rightarrow \Psi$ we have  $\bigoplus_{i=1}^{k+l} \sigma_i' \models \Psi$ and can thus take $\sigma_{k+1}', \ldots, \sigma_{k+l}'$ as the desired final states to show the validity of  $\rel{\Phi}{\univC}{\existsC}{\Psi}$.
		
		\item \nameref{rule:forall-if-app}:
		By IH we get that $\rel{\Phi \land b}{\prog_1 \mysemi \prog_3 \oast\univC}{\existsC}{\Psi}$ and $\rel{\Phi \land \neg b}{\prog_2 \mysemi \prog_3 \oast\univC}{\existsC}{\Psi}$ are valid. 
		We show that $\rel{\Phi}{\myif(b, \prog_1, \prog_2) \mysemi \prog_3 \oast\univC}{\existsC}{\Psi}$ is valid. 
		Let $\sigma_1, \ldots, \sigma_{k+l}$ be states such that $\bigoplus_{i=1}^{k+l} \sigma_i \models \Phi$ and consider any final states $\sigma_1', \ldots, \sigma_k'$ such that $\sem{\prog_i}(\sigma_i, \sigma_i')$ for all $i \in [1, k]$.
		Now either $\bigoplus_{i=1}^{k+l} \sigma_i \models \Phi \land b$ or $\bigoplus_{i=1}^{k+l} \sigma_i \models \Phi \land \neg b$.
		In the first case former case, we can use the witness final states given by the validity of $\rel{\Phi \land b}{\prog_1 \mysemi \prog_3 \oast\univC}{\existsC}{\Psi}$ to show the validity of $\rel{\Phi}{\myif(b, \prog_1, \prog_2) \mysemi \prog_3 \oast\univC}{\existsC}{\Psi}$.
		Note,  for any state $\sigma \models b$ we get that $\sem{\myif(b, \prog_1, \prog_2)\mysemi \prog_3}(\sigma, \sigma')$ iff $\sem{\prog_1\mysemi \prog_3}(\sigma, \sigma')$ for any state $\sigma'$.
		The case where $\bigoplus_{i=1}^{k+l} \sigma_i \models \Phi \land \neg b$ is analogous. 
		
		\item \nameref{rule:forall-step-app}:
		By IH we get that $\rel{\Phi'}{\prog_2 \oast \univC}{\existsC}{\Psi}$ is valid, and by the assumption that $\vdash \tripH{\cdot}{\cdot}{\cdot}$ is sound, the HT $\tripH{\Phi}{\prog_1}{\Phi'}$ is valid. 
		We show that $\rel{\Phi}{\prog_1 \mysemi \prog_2 \oast \univC}{\existsC}{\Psi}$ is valid. 
		Let $\sigma_1, \ldots, \sigma_{k+l}$ be states such that $\bigoplus_{i=1}^{k+l} \sigma_i \models \Phi$ and any final states $\sigma_1', \ldots, \sigma_k'$ such that $\sem{\prog_i}(\sigma_i, \sigma_i')$ for all $i \in [1, k]$.
		In particular $\sem{\prog_1\mysemi \prog_2}(\sigma_1, \sigma_1')$.
		By the semantics there thus exists a state $\sigma_1''$ such that $\sem{\prog_1}(\sigma_1, \sigma_1'')$ and $\sem{\prog_2}(\sigma_1'', \sigma_1')$. 
		As the HT $\tripH{\Phi}{\prog_1}{\Phi'}$ is valid, we have that $\sigma_1'' \models \Phi'$. 
		So $\sigma_1'' \oplus \bigoplus_{i=2}^{k+l} \sigma_i \models \Phi'$ (note that $\prog_1$ does not manipulate any variables in $\sigma_2, \ldots, \sigma_{k+l}$). 
		We can thus use the witnessing final states provided by the validity of $\rel{\Phi'}{\prog_2 \oast \univC}{\existsC}{\Psi}$ to show that $\rel{\Phi}{\prog_1 \mysemi \prog_2 \oast \univC}{\existsC}{\Psi}$ is valid.
		
		\item \nameref{rule:forall-assume-app}:
		We show $\rel{\Phi}{\myassume(b) \mysemi \prog \oast \univC}{\existsC}{\Psi}$ is valid.
		For this, let $\sigma_1, \ldots, \sigma_{k+l}$ be states such that $\bigoplus_{i=1}^{k+l} \sigma_i \models \Phi$ and consider arbitrary final states $\sigma_1', \ldots, \sigma_k'$ such that $\sem{\prog_i}(\sigma_i, \sigma_i')$ for all $i \in [1, k]$.
		In particular, $\sem{\myassume(b)\mysemi \prog}(\sigma_1, \sigma_1')$ so (by the semantics of \myassume) $\sigma_1 \models b$. 
		For any state $\sigma$ with $\sigma \models b$ we have $\sem{\myassume(b)\mysemi \prog}(\sigma, \sigma')$ iff $\sem{\prog}(\sigma, \sigma')$ for any state $\sigma'$.
		We can thus use the witness final states given by the validity of $\rel{\Phi \land b}{\prog \oast \univC}{\existsC}{\Psi}$ to show that $\rel{\Phi}{\myassume(b) \mysemi \prog \oast \univC}{\existsC}{\Psi}$ is valid.
		
		\item \nameref{rule:exists-assume-app}:
		We show $\rel{\Phi}{\univC}{\myassume(b) \mysemi \prog \oast \existsC}{\Psi}$ is valid.
		For this, let $\sigma_1, \ldots, \sigma_{k+l}$ be states such that $\bigoplus_{i=1}^{k+l} \sigma_i \models \Phi$ and consider arbitrary final states $\sigma_1', \ldots, \sigma_k'$ such that $\sem{\prog_i}(\sigma_i, \sigma_i')$ for all $i \in [1, k]$.
		By assumption we have $\Phi \Rightarrow b$, so $\sigma_{k+1} \models b$.
		For any state $\sigma$ with $\sigma \models b$ we have $\sem{\myassume(b)\mysemi \prog}(\sigma, \sigma')$ iff $\sem{\prog}(\sigma, \sigma')$ for any state $\sigma'$.
		As $\sigma_{k+1} \models b$, we can use the final states given by the validity of $\rel{\Phi}{\univC}{\prog \oast \existsC}{\Psi}$ to show that $\rel{\Phi}{\univC}{\myassume(b) \mysemi \prog \oast \existsC}{\Psi}$ is valid.
		
		\item \nameref{rule:forall-inf-nd-app}:
		We show that $\rel{\Phi}{x \myassign \star \mysemi \prog \oast \univC}{\existsC}{\Psi}$ is valid. 
		For this, let $\sigma_1, \ldots, \sigma_{k+l}$ be states such that $\bigoplus_{i=1}^{k+l} \sigma_i \models \Phi$ and let $\sigma_1', \ldots, \sigma_k'$ be any final states such that $\sem{\prog_i}(\sigma_i, \sigma_i')$ for all $i \in [1, k]$.
		In particular, $\sem{x \myassign \star \mysemi \prog}(\sigma_1, \sigma_1')$, so by the semantics there exists a state $\sigma_1''$ such that $\sem{x \myassign \star}(\sigma_1, \sigma_1'')$ and $\sem{\prog}(\sigma_1'', \sigma_1')$.
		In particular (by the semantics of $x \myassign \star$), we get that $\sigma_1'' = \sigma_1[x \mapsto z]$ for some $z \in \mathbb{Z}$.
		Now $\sigma_1 \oplus \bigoplus_{i=2}^{k+l} \sigma_i \models \Phi$ so $\sigma_1'' \oplus \bigoplus_{i=2}^{k+l} \sigma_i \models \exists x \ldot\Phi$.
		We can thus use the final states given by the validity of $\rel{\exists x \ldot \Phi}{\prog \oast \univC}{\existsC}{\Psi}$ to show that $\rel{\Phi}{x \myassign \star \mysemi \prog \oast \univC}{\existsC}{\Psi}$ is valid.
		
		\item \nameref{rule:exists-inf-nd-app}:
		We show that $\rel{\Phi}{\univC}{x \myassign \star \mysemi \prog \oast \existsC}{\Psi}$ is valid. 
		For this, let $\sigma_1, \ldots, \sigma_{k+l}$ be states such that $\bigoplus_{i=1}^{k+l} \sigma_i \models \Phi$ and let $\sigma_1', \ldots, \sigma_k'$ be any final states such that $\sem{\prog_i}(\sigma_i, \sigma_i')$ for all $i \in [1, k]$.
		Define $\sigma_{k+1}'' := \sigma_{k+1}[x \mapsto \bsem{e}(\sigma_{k+1})]$ (where $e$ is the expression used in the rule application).
		As $\bigoplus_{i=1}^{k+l} \sigma_i \models \Phi$ we have that $\sigma_{k+1}'' \oplus \bigoplus_{i=1 \mid i \neq k+1}^{k+l} \sigma_i \models (\exists x. \Phi) \land x = e$.
		By IH we can thus find final states $\sigma_{k+1}', \ldots \sigma_{k+l}'$ such that $\bigoplus_{i=1}^{k+l} \sigma_i' \models \Psi$.
		In particular, we get that $\sem{\prog}(\sigma_{k+1}'', \sigma_{k+1}')$ and thus $\sem{x \myassign \star \mysemi \prog}(\sigma_{k+1}, \sigma_{k+1}')$. 
		We can thus use $\sigma_{k+1}', \ldots \sigma_{k+l}'$ as witnesses to show the validty of $\rel{\Phi}{\univC}{x \myassign \star \mysemi \prog \oast \existsC}{\Psi}$ as required. \qed
	\end{itemize}
\end{proof}

\subsection{Completeness}\label{app:fehl_compltness}

In this subsection we prove FEHL complete. 
Similar to \cite{NagasamudramN21}, we show completeness by giving a single rule that encodes the composition of all programs. 
Consider the following rule:
\begin{center}
	\scalebox{0.9}{\parbox{\linewidth}{
			\begin{topprooftree}{\labelText{\texttt{(Self-Composition)}}{rule:comp}}
				\AxiomC{$\Phi = \theta_1, \ldots, \theta_{k+l+1} = \Psi$}
				\AxiomC{$\Big[\vdash \tripH{\Theta_j}{\prog_j}{\Theta_{j+1}}\Big]_{j = 1}^k$}
				\AxiomC{$\Big[\vdash \tripN{\Theta_j}{\prog_j}{\Theta_{j+1}}\Big]_{j = k+1}^{k+l}$}
				\TrinaryInfC{$\vdash \rel{\Phi}{\prog_1 \oast \cdots \oast \prog_k}{\prog_{k+1} \oast \cdots \oast \prog_{k+l}}{\Psi}$}
			\end{topprooftree}
	}}
\end{center}
That is, we require $k+l+1$ assertions $\theta_1, \ldots, \theta_{k+l+1}$, such that $\theta_1 = \Phi$ and $\theta_{k+l+1} = \Psi$. 
We then iteratively step through programs $\prog_1, \ldots, \prog_k$ using HTs, followed by an analysis of $\prog_{k+1}, \ldots, \prog_{k+l}$ using UHTs. 

\begin{proposition}\label{lem:comp_comp}
	Assume that $\vdash \tripH{\cdot}{\cdot}{\cdot}$ and $\vdash \tripN{\cdot}{\cdot}{\cdot}$ are complete proof systems for HTs and UHTs, respectively. 
	The proof system consisting only of \nameref{rule:comp} is complete for FEHTs.
\end{proposition}
\begin{proof}
	Assume that $\vdash \rel{\Phi}{\prog_1 \oast \cdots \oast \prog_k}{\prog_{k+1} \oast \cdots \oast \prog_{k+l}}{\Psi}$ is valid. 
	We show that the FEHT $\rel{\Phi}{\prog_1 \oast \cdots \oast \prog_k}{\prog_{k+1} \oast \cdots \oast \prog_{k+l}}{\Psi}$ is derivable using \nameref{rule:comp}.
	This derivation is similar to the proof in the case of $k$-safety \cite[Proposition 9]{NagasamudramN21}:
	For $j$ from $1$ to $k$, we define $\theta_{j+1}$ as all states reachable by executing $\prog_j$ from $\theta_j$. 
	Now as $\vdash \rel{\Phi}{\prog_1 \oast \cdots \oast \prog_k}{\prog_{k+1} \oast \cdots \oast \prog_{k+l}}{\Psi}$ is valid, from any state in $\theta_{k+1}$ we can find some execution of $\prog_{k+1}, \ldots, \prog_{k+l}$ that end in some relational state in $\Psi$. 
	For $j$ from $k+1$ to $k+l$, we can thus define $\theta_{j+1}$ by, for any state $\sigma \models \theta_j$ adding a state $\sigma'$ to $\theta_{j+1}$ that results from $\prog_j$'s selected execution on $\sigma$.
	The statement then follows from the fact that $\vdash \tripH{\cdot}{\cdot}{\cdot}$ and $\vdash \tripN{\cdot}{\cdot}{\cdot}$ are complete. \qed
\end{proof}

\completeness*
\begin{proof}
	By \Cref{lem:comp_comp}, proof rule \nameref{rule:comp} is complete for FEHTs. 
	All that remains to argue is that we can derive \nameref{rule:comp} within FEHL.	
	We can easily do this by bring each program $\prog$ in $\univC \cup \existsC$ into the form $\prog\mysemi \myskip$ (using \nameref{rule:forall-intro} and \nameref{rule:exists-intro}), apply \nameref{rule:forall-step-app} $k$ times to all the universally quantified copies, followed by $l$ applications of  \nameref{rule:exists-step-app} and always transform the precondition to $\Theta_j$ (after $j$ applications).
	Afterward we are left with precondition $\theta_{k+l+1} = \Psi$ and  $k+l$ \myskip{} programs, which discharge using $k$ applications of \nameref{rule:forall-elim}, $l$ applications of \nameref{rule:exists-elim}, and a final application of \nameref{rule:done}.\qed
\end{proof}

\section{Underapproximate Hoare Triples}\label{app:ehl}

We elaborate briefly on the (forward-style) underapproximate Hoare Triples (UHTs) we use to discharge non-relation obligations for existentially quantified executions. 
As already defined in \Cref{sec:FEHL}:

\begin{definition}
	An UHT $\tripN{\Phi}{\prog}{\Psi}$ is valid if for all states $\sigma$ with $\sigma \models \Phi$ there \emph{exists} a state $\sigma'$ such that $\sem{\prog}(\sigma, \sigma')$ and $\sigma' \models \Psi$.
\end{definition}

We can show that -- similar to HTs and ITs -- UHTs are supported by a sound-and-complete proof system, as needed for completeness of FEHL (cf.~\Cref{theo:completeness}).
As usual for complete proof systems, our system does not provide a direct verification path (deciding if an UHT is valid is undecidable), but rather strengthens FEHLs completeness statement. 
We present our proof system in \Cref{fig:ehl_rules}.
Most rules are standard.
Note that our loop rule is similar to that found in total Hoare logic \cite{Apt81}, i.e., implicitly encoded a variant. 
A simple induction shows:

\begin{proposition}[Soundness]
	If $\vdash \tripN{\Phi}{\prog}{\Psi}$ then $\tripN{\Phi}{\prog}{\Psi}$ is valid.
\end{proposition}

\begin{proposition}[Completeness]
	If $\tripN{\Phi}{\prog}{\Psi}$ is valid then $\vdash \tripN{\Phi}{\prog}{\Psi}$.
\end{proposition}

\begin{figure}[!t]
	
	\def\defaultHypSeparation{\hskip .1in}
	
	\begin{minipage}{0.24\textwidth}
		\centering
		\scalebox{0.9}{\parbox{\linewidth}{
		\begin{prooftree}
			\AxiomC{}
			\UnaryInfC{$\vdash\tripN{\Phi}{\myskip}{\Phi}$}
		\end{prooftree}}}
	\end{minipage}%
	\begin{minipage}{0.25\textwidth}
		\centering
		\scalebox{0.9}{\parbox{\linewidth}{
		\begin{prooftree}
			\AxiomC{}
			\UnaryInfC{$\vdash\tripN{\Phi[e/x]}{x \myassign e}{\Phi}$}
		\end{prooftree}}}
	\end{minipage}%
	\begin{minipage}{0.27\textwidth}
		\centering
		\scalebox{0.9}{\parbox{\linewidth}{
		\begin{prooftree}
			\AxiomC{}
			\UnaryInfC{$\vdash\tripN{\exists x. \Phi}{x \myassign \star}{\Phi}$}
		\end{prooftree}}}
	\end{minipage}%
	\begin{minipage}{0.27\textwidth}
		\centering
		\scalebox{0.9}{\parbox{\linewidth}{
		\begin{prooftree}
			\AxiomC{$\vdash\Phi \Rightarrow b$}
			\UnaryInfC{$\vdash\tripN{\Phi}{\myassume(b)}{\Phi}$}
		\end{prooftree}}}
	\end{minipage}
	
	\vspace{-2mm}
	
	\begin{minipage}{0.45\textwidth}
		\centering
		\scalebox{0.9}{\parbox{\linewidth}{
		\begin{prooftree}
			\AxiomC{$\vdash\tripN{\Phi \land b}{\prog_1}{\Psi}$}
			\AxiomC{$\vdash\tripN{\Phi \land \neg b}{\prog_2}{\Psi}$}
			\BinaryInfC{$\vdash\tripN{\Phi}{\myif(b, \prog_1, \prog_2)}{\Psi}$}
		\end{prooftree}}}
	\end{minipage}%
	\begin{minipage}{0.3\textwidth}
		\centering
		\scalebox{0.9}{\parbox{\linewidth}{
		\begin{prooftree}
			\AxiomC{\stackanchor{$\Phi \Rightarrow \Phi'$}{$\Psi' \Rightarrow \Psi$}}
			\AxiomC{$\vdash\tripN{\Phi'}{\prog}{\Psi'}$}
			\BinaryInfC{$\vdash\tripN{\Phi}{\prog}{\Psi}$}
		\end{prooftree}}}
	\end{minipage}%
	\begin{minipage}{0.25\textwidth}
		\centering
		\scalebox{0.9}{\parbox{\linewidth}{
		\begin{prooftree}
			\AxiomC{$\Big[\vdash\tripN{\Phi_i}{\prog}{\Psi}\Big]_{i=1}^n$}
			\UnaryInfC{$\vdash\tripN{\bigvee_{i=1}^n \Phi_i}{\prog}{\Psi}$}
		\end{prooftree}}}
	\end{minipage}%

	\vspace{-2mm}

	\begin{minipage}{0.6\textwidth}
		\centering
		\scalebox{0.9}{\parbox{\linewidth}{
		\begin{prooftree}
			\AxiomC{$\vdash\tripN{\Phi_{i+1}}{\prog}{\Phi_{i}}$}
			\AxiomC{$\forall i  \geq 1\ldot \Phi_i \Rightarrow b$}
			\AxiomC{$\Phi_0 \Rightarrow \neg b$}
			\TrinaryInfC{$\vdash\tripN{\exists n \in \nat. \Phi_n}{\mywhile(b, \prog)}{\Phi_0}$}
		\end{prooftree}}}
	\end{minipage}%
	\begin{minipage}{0.3\textwidth}
		\centering
		\scalebox{0.9}{\parbox{\linewidth}{
		\begin{prooftree}
			\AxiomC{$\vdash\tripN{\Phi}{\prog_1}{\Phi'}$}
			\AxiomC{$\vdash\tripN{\Phi'}{\prog_2}{\Psi}$}
			\BinaryInfC{$\vdash\tripN{\Phi}{\prog_1\mysemi\prog_2}{\Psi}$}
		\end{prooftree}}}
	\end{minipage}%
	
	\caption{Proof rules for UHTs } \label{fig:ehl_rules}
\end{figure}

\section{Parametric Postconditions}

\paraPostSound*
\begin{proof}
	We assume that 
	\begin{align}\label{eq:FEF}
		\textstyle
		\bigforall_{x \in \progVars_1 \cup \cdots \cup \progVars_k} x\ldot \bigexists_{\mu \in \parameter} \mu\ldot \condF \land \bigforall_{x \in \progVars_{k+1} \cup \cdots \cup \progVars_{k+l}} x \ldot  (\postF \Rightarrow \Psi)
	\end{align}
	holds and show that $\rel{\Phi}{\prog_1 \oast \cdots \oast \prog_k}{\prog_{k+1} \oast \cdots \oast \prog_{k+l}}{\Psi}$ is valid. 
	Let $\sigma_1, \ldots, \sigma_{k+l}$ and $\sigma_1', \ldots, \sigma_{k}'$ be arbitrary states such that  $\bigoplus_{i=1}^{k+l} \sigma_i \models \Phi$ and $\sem{\prog_i}(\sigma_i, \sigma_i')$ holds for all $i \in [1, k]$.
	We instantiate the universal quantifiers in \Cref{eq:FEF} with the concrete values from $\sigma_1', \ldots, \sigma_{k}'$.
	As \Cref{eq:FEF} holds, we thus get a parameter evaluation $\kappa$ such that $\kappa \models \condF$ (by extracting a witness for the existential quantifiers in \Cref{eq:FEF}).
	We now define final states $\sigma_{k+1}', \ldots, \sigma_{k+l}'$ as some states such that $\bigoplus_{i=1}^{k+l} \sigma_i'  \models \postF[\kappa]$. 
	By condition \textbf{(1)} in the definition of a parametric postcondition (\Cref{def:paraPost}) such states exist. 
	Moreover, by condition \textbf{(2)} in \Cref{def:paraPost}, we have that  $\sem{\prog_i}(\sigma_i, \sigma_i')$ holds for all $i \in [k+1, k+l]$.
	We now instantiate the innermost universal quantification in \Cref{eq:FEF} with the concrete values from $\sigma_{k+1}', \ldots, \sigma_{k+l}'$. 
	By assumption, $\bigoplus_{i=1}^{k+l} \sigma_i' \models \postF[\kappa]$, so as the premise of the implication in \Cref{eq:FEF} holds.
	We thus get $\bigoplus_{i=1}^{k+l} \sigma_i' \models \Psi$ as required. 
	The final states $\sigma_{k+1}', \ldots, \sigma_{k+l}'$ thus serve as witnesses to show the validity of $\rel{\Phi}{\prog_1 \oast \cdots \oast \prog_k}{\prog_{k+1} \oast \cdots \oast \prog_{k+l}}{\Psi}$.\qed
\end{proof}

\fi

\end{document}